\documentclass[11pt]{article}
\usepackage[left=2.5cm,top=2.5cm, bottom=1.2in, right=2.5cm]{geometry}
\date{}
\usepackage{amsfonts}
\usepackage{amsmath, amsthm, amssymb}
\usepackage{graphicx}
\usepackage{hyperref, comment, cite}
\usepackage{tikz}

\usetikzlibrary{fit}					% fitting shapes to coordinates
\usetikzlibrary{backgrounds}	

\newtheorem{theorem}{Theorem}%[section]
\newtheorem{lemma}[theorem]{Lemma}

\newtheorem{definition}[theorem]{Definition}
\newtheorem{example}[theorem]{Example}

\definecolor{darkblue}{rgb}{0, .07, .5}
\definecolor{darkred}{rgb}{0.5,0,0}
 \definecolor{mahogany}{rgb}{0.65, 0., 0.5}

\newcommand{\fR}{\mathfrak{R}}

\newcommand{\Var}{\text{\rm{Var}}}

\newcommand{\PR}{\text{\rm{PR}}}

\newcommand{\E}{\mathbb{E}}

\begin{document}
\title{On the Duality of Additivity and Tensorization}

\author{ Salman Beigi$^1$, Amin Gohari$^{1,2}$\\
$^1${\it \small School of Mathematics,} {\it \small Institute for Research in Fundamental Sciences (IPM), Tehran, Iran}\\
$^2${\it \small Department of Electrical Engineering,} {\it \small Sharif University of Technology, Tehran, Iran}}

\maketitle

\begin{abstract}
A function is said to be additive if, similar to mutual information, expands by a factor of $n$, when evaluated on $n$ i.i.d.\ repetitions of a source or channel. On the other hand, a function is said to satisfy the tensorization property if it remains unchanged when evaluated on i.i.d.\ repetitions. Additive rate regions are of fundamental importance in network information theory, serving as capacity regions or upper bounds thereof. 
Tensorizing measures of correlation have also found applications in distributed source and channel coding problems as well as the distribution simulation problem. Prior to our work only two measures of correlation, namely the hypercontractivity ribbon and maximal correlation (and their derivatives), were known to have the tensorization property. 
In this paper, we provide a general framework to obtain a region with the tensorization property from any additive rate region.
We observe that hypercontractivity ribbon indeed comes from the dual of the rate region of the Gray-Wyner source coding problem, and generalize it to the multipartite case. Then we define other measures of correlation with similar properties from other source coding problems. 
We also present some applications of our results. 
\end{abstract}

%*****************************Introduction*************************
\section{Introduction}
Additivity is a fundamental property of interest in information theory (e.g., see \cite{Holevo, GeneralAdditivity}) since capacity regions by their operational definition are  additive for product of identical channels or sources. Tensorization is another important property of regions in information theory which in this paper we interpret as the dual of additivity problem. Let us explain the notions of additvity and tensorization via the example of non-interactive distribution simulation \cite{KamathAnantharam}.

Fix some bipartite  distribution $p_{XY}$. Suppose that two parties, Alice and Bob, are given i.i.d.\ samples $X^n$ and $Y^n$ respectively, and they are asked to output $A$ and $B$ (respectively) distributed according to some predetermined  $q_{AB}$. Alice and Bob can choose $n$ to be as large as they want, but are not allowed to communicate. The problem of deciding whether this task is doable or not is a hard problem in general. Nevertheless, we may obtain impossibility results using the data processing inequality.  

Suppose that $I(X^n; Y^n) < I(A; B)$. In this case by the data processing inequality local transformation of $(X^n, Y^n)$ to $(A, B)$ is infeasible. However, note that mutual information is \emph{additive}, i.e., we have $I(X^n; Y^n) = n\cdot I(X; Y)$. Then, unless $X$ and $Y$ are independent,  by choosing $n$ to be large enough, $I(X^n; Y^n)$ becomes as large as we want and greater than $I(A; B)$. Therefore, the data processing inequality of mutual information does not give us any useful bound on this problem, simply because mutual information is additive. 
 
Now suppose that there is some function $\rho(\cdot, \cdot)$ of bipartite distributions that similar to mutual information satisfies the data processing inequality, but is not additive. More precisely, suppose that
$$\rho(X^n, Y^n) = \rho(X, Y).$$
That is, $\rho(\cdot, \cdot)$ extremely violates additivity and satisfies the above equation which is called the \emph{tensorization} property. Given such a measure and following the previous argument we find that local transformation of $(X^n, Y^n)$ to $(A, B)$ is impossible (even for arbitrarily large $n$) if $\rho(X, Y)< \rho(A, B)$.

In the above example we see how tensorization naturally appears as a tool to solve information theoretic problems. In the following by giving some examples we clarify the notions of additivity and tensorization and then explain our results.

\subsection{Additivity} 
Capacity regions by their operational definition are  additive for product of identical channels or sources
since they are expressed as a limit of multi-letter instances of the problem as the blocklength  goes  to infinity. For instance, consider the capacity of a point to point channel:
$$\mathcal{C}(p(y|x))=\max_{p(x)}I(X;Y).$$
By its  operational definition, the capacity of a product of identical channels is equal to the sum of the capacities of the individual channels 
$$\mathcal{C}(p(y_1|x_1)p(y_2|x_2))=\mathcal{C}(p(y_1|x_1))+\mathcal{C}(p(y_2|x_2)).$$
This is called the additivity property of the channel capacity.

Defining additivity for general network information theory problems, involving relay and feedback is more involved~\cite{GeneralAdditivity}, but for one-hop networks, when we are dealing with a rate region $\mathcal{R}(\cdot)$, we say that it is additive if 
\begin{align}\label{eq:add-general}
\mathcal{R}(p\times p)=\mathcal{R}(p)+\mathcal{R}(p),
\end{align} where $p$ is the underlying channel or joint distribution and $+$ is the Minkowski sum (point-wise sum). 
%For instant, in the distribution source coding problem of Slepian and Wolf, for a given $p_{XY}$ its rate region $\mathcal R(p)$ consists of certain pairs $(R_X, R_Y)$. Then by its operation definition, $\mathcal R(p\times p)$, the rate region associated to two copies of $p_{XY}$ consists of pairs $(R_X+ R'_X, R_Y+ R'_Y)$ such that $(R_X, R_Y), (R'_X, R'_Y)\in \mathcal R(p)$. 

Additive regions are of fundamental importance to network information theory, not only because of the additivity of capacity regions, but also because the known upper bounds on capacity regions are additive. 

\subsection{Tensorization}
Tensorization has received relatively less attention comparing to additivity. The simplest example to illustrate the definition and applications of tensorization  is via Witsenhausen's extension~\cite{Witsenhausen} of the G\'acs-K\"orner common information \cite{GacsKorner}. Assume that Alice and Bob are observing i.i.d.\ repetitions of random variables $X^n$ and $Y^n$. Their goal is to extract common randomness via functions $f(\cdot)$ and $g(\cdot)$ such that with high probability $f(X^n)=g(Y^n)$. G\'acs and K\"orner show that unless $X=(C, X')$ and $Y=(C, Y')$ for some explicit common part $C$, the rate of common randomness extraction is zero. This result was strengthened  by Witsenhausen, who showed that if $X$ and $Y$ do not have any explicit common part, it is not possible for Alice and Bob to extract even a single common random bit. This was shown by utilizing a measure of correlation, called the \emph{maximal correlation}~\cite{Hirschfeld, Gebelein, Renyi1, Renyi2, Witsenhausen}. 

Maximal correlation of  a given bipartite probability distribution $p_{XY}$ is the maximum of Pearson's correlation coefficient over all functions of $X$ and $Y$, i.e.,
\begin{align}\label{eq:max-correlatoin-31}
\rho(X, Y)=& \max \frac{\E\big[(f_X-\E[f_X])(g_Y-\E[g_Y])\big]}{\sqrt{\Var[f_X]\Var[g_Y]}},
\end{align}
where $\E[\cdot]$ and $\Var[\cdot]$ are expectation value and variance respectively. Moreover, the maximum is taken over all non-constant functions $f_X, g_Y$ of $X$ and $Y$ respectively. Maximal correlation can equivalently be written as
\begin{align*}
\rho(X, Y) =  \max  ~& \E[f_{X}\,g_{Y}]\\
&\E_{X} [f] = \E_{Y} [g]=0, \\
& \E [f^2]= \E [g^2] =1.
\end{align*}

We always have $0\leq \rho(X, Y)\leq 1$. Moreover, $\rho(X, Y)=0$ if and only if $X$ and $Y$ are independent, and $\rho(X, Y)=1$ if and only if $X$ and $Y$ have an explicit common data as defined above~\cite{Witsenhausen}. Maximal correlation has the following two properties: 
\begin{itemize}
\item \emph{Tensorization:} We have
\begin{align}\rho(XX', YY') = \max\{\rho(X, Y), \rho(X', Y')\},\label{eqn:amnewtensor}\end{align}
when $XY$ and $X'Y'$ are independent, i.e., $p_{XX'YY'}=p_{XY}\cdot p_{X'Y'}$. 
\item \emph{Data Processing:} We have
\begin{align}\rho(X', Y') \leq \rho(X, Y),\label{eqn:dataprocessingrho}\end{align}
when $X'\rightarrow X\rightarrow Y\rightarrow Y'$ forms a Makov chain. Thus maximal correlation can be thought of as a measure of correlation
\end{itemize}

Applying the above two properties to the G\'acs-K\"orner problem we find that 
$$\rho(f(X^n), g(Y^n))\leq \rho(X^n, Y^n)=\rho(X, Y).$$ 
As a result, if $\rho(X, Y)<1$, then $\rho(f(X^n), g(Y^n))$ will also be strictly less than one. Then Witsenhausen's result is obtained using a certain continuity of maximal correlation and the fact that the maximal correlation of two perfectly correlated bits is $1$.

More generally, the tensorization and data processing properties of maximal correlation imply some bounds on the problem of non-interactive distribution simulation discussed above. That is, if we generate random variables  $A$ and $B$ from $n$ i.i.d.\ repetitions of $X$ and $Y$ respectively, i.e., if $A\rightarrow X^n \rightarrow Y^n\rightarrow B$ for some $n$, then 
\begin{align}
\rho(A,B)\leq \rho(X,Y).
\label{eqn:rho-tot}\end{align}

%The above relation has applications in non-interactive distribution simulation \cite{KamathAnantharam}. In this problem, two parties are assumed to hold arbitrary many copies of random variables $X$ and $Y$ as resources. They are interested to generate $X'$ and $Y'$ with a desired joint distribution $p_{X'Y'}$ from copies of $X^n$ and $Y^n$ by  local operations. This is possible only if equation \eqref{eqn:rho-tot} holds. Observe that the usual data processing inequality via mutual information  gives  $I(X';Y')\leq I(X^n; Y^n)=nI(X;Y)$, which does not yield any useful bound, since $n$ can be arbitrarily large. 

Tensorization is also helpful in distributed source and channel coding problems \cite{KangUlukus}. For instance, consider the problem of transmission of correlated sources over a MAC channel. Assuming that the correlated sources observed by the two transmitters are i.i.d.\ repetitions of $(A, B)$, their inputs to the MAC channel at time $i$ which we denote by $X_{i}$ and $Y_{i}$ satisfy $X_{i}\rightarrow A^n\rightarrow B^n\rightarrow Y_{i}$, and hence we must have $\rho(X_{i}, Y_{i})\leq \rho(A,B)$. Therefore, the set of possible input distributions to the MAC is restricted. This can be used to prove impossibility results in transmission of correlated sources.

In general, if $\Upsilon(p)$ is a region for a given distribution $p$, we say that it \emph{tensoizes} or has the tensorization property if 
\begin{align}\label{eq:def-tens-region}
\Upsilon(p_1\times p_2)=\Upsilon(p_1)\cap \Upsilon(p_2),
\end{align}
for any $p_1, p_2$. This in particular implies that for i.i.d.\ repetitions $p^n$ we have
\begin{align}\label{eq:weak-tensor}
\Upsilon(p^n) = \Upsilon(p).
\end{align}
Equation \eqref{eq:weak-tensor} is a weaker version of~\eqref{eq:def-tens-region}, and is called \emph{weak tensorization} property. In this paper we mostly consider this weak tensorization. So when we say tensorization, we mean~\eqref{eq:weak-tensor} unless stated otherwise.  If $\Upsilon(p)$ is a scalar (as for maximal correlation),  tensorization translates to
$$\Upsilon(p_1\times p_2)=\max\{\Upsilon(p_1),\Upsilon(p_2)\}.$$ 

Tensorizing regions serve as \emph{measures of correlation} if they satisfy an additional data processing inequality. Only two examples of tensorizing regions that satisfy the data processing inequality are known in the literature, and the other such measures are derived from these two. One of them is the hypercontractivity ribbon~\cite{AhlswedeGacs}. The other one is a generalization of maximal correlation called maximal correlation ribbon~\cite{OurPaper}. Both hypercontractivity ribbon and maximal correlation ribbon are subsets of the real plane and satisfy~\eqref{eq:def-tens-region}.

\subsection{Our contributions} 
The key idea is that given a region $\mathcal R$ that is an additive  function of the joint distribution $p(x_1, \dots, x_k)$, the cone at which $\mathcal{R}$    is seen from zero is a tensorizable function of the joint distribution. Furthermore, by subtracting any additive vector from $\mathcal{R}$ the above statement extends to cones at which $\mathcal{R}$ seen from one of its corners. This allows for 
introducing new measures of correlation that (weakly) tensorize. Our new measures are defined as the dual of the rate regions of certain source coding problems. Since by its operational definition, the source coding capacity region is additive, we get an operational proof of the tensorization property. Moreover, the source coding problems that we consider involve private links to the receivers, making it possible to use the Slepian-Wolf theorem to transmit parts of the sources through these links. We show that this implies the data processing property in the dual region. The operational proof of data processing does not rely on knowing the exact characterization of the original problem (in terms of mutual information). 

With this approach 
we define new regions that tensorize and satisfy the data processing inequality. In fact, we show that hypercontractivity ribbon and maximal correlation are simply two members of a larger class of regions with the above properties. In particular, making connections with the Gray-Wyner source coding problem, we naturally extend the definition of the hypercontractivity ribbon to the multipartite setting. 
Our construction also generalizes the technique of \emph{initial efficiency} to produce tensorizing regions from additive ones (see \cite{CuffInitial, CuffInitial2}).

\subsection{Structure of the paper} 
This paper is organized as follows. In Section \ref{sec:addtotens} we discuss how one can get tensorizing regions from additive ones. This is followed by a series of examples in Sections \ref{sec:losslesshleper},  \ref{sec:ex2}, \ref{sec:fork} and \ref{sec:cond:HC}, where new multipartite and conditional regions are defined. Section \ref{sec:computing} addresses the difficulty of computing regions based on auxiliary random variables, and provides an approach for finding alternative local regions that are easier to compute. Section \ref{sec:two-way} discusses additivity and tensorization for a two-way channel problem, and its application in simulating a two-way channel from another. 

\subsection{Notation} 
We mainly adopt the notation of \cite{elgamal}. In particular, we use  $[k]$ to denote the set $\{1,2,\dots, k\}$. We use $x_{[k]}$ to denote the sequence $(x_1, x_2, \dots, x_k)$, and $x_{[k]}^n$ to denote $(x_1^n, x_2^n, \dots, x_k^n)$ where $x_i^n=(x_{i1}, x_{i2}, \dots, x_{in})$. In general, for a subset $T$ by $x_T$ we mean the tuple of $x_i$'s for $i\in T$. The complement of subset $T$ is denoted by $T^c$.
Random variables are shown in capital letters, whereas their realizations are shown using the lowercase letters.

Expectation value and variance are respectively denoted by $\E[\cdot]$ and $\Var[\cdot]$. When expectation is computed with respect to some distribution $p(x)$ with associated random variable $X$, we sometimes denote $\E[\cdot]$ by $\E_{X}[\cdot]$. We adopt the same notation for variance too. 

Letting $p(x,y)$ be some bipartite distribution, the conditional expectation $\E_{X|Y}[\cdot]$ gives a function of $Y$ which itself is a random variable. We sometimes denote this conditional expectation by $\E[\,\cdot\, | Y]$.

The set of real numbers is denoted by $\mathbb R$, and $\mathbb R_+=[0, \infty)$ denotes the set of non-negative real numbers.

%************************************************************************************
\section{From additivity to tensorization}\label{sec:addtotens}
Consider an arbitrary source coding problem, involving i.i.d.\ repetitions of random variables $(X_1, \dots, X_k)$, with some capacity rate\footnote{The region $\mathcal{R}$ depends on the joint distribution $p(x_1, \dots, x_k)$ but we adopt the common abuse of notation in information theory to write it as $\mathcal{R}(X_1, \dots, X_k)$.} region $\mathcal{R}(X_1, \dots, X_k)$ consisting of rate tuples $(R_1, \dots, R_m)$. The definition of the source coding problem can be quite arbitrary; we only use the fact that from the operational definition of the rate region we have
\begin{align}
(R_1, \dots, R_m)\in \mathcal{R}(X_1, \dots, X_k)  \quad\Longleftrightarrow \quad (nR_1, \dots, nR_m)\in \mathcal{R}(X_1^n,  \dots, X_k^n)
\label{addveq},
\end{align}
where $(X_1^n, X_2^n, \dots, X_k^n)$ is $n$ i.i.d.\ repetitions of $(X_1, X_2, \dots, X_k)$.

Let $\lambda_i$ for $i\in[m]$, and $ \theta_S$ for non-empty subsets $S\subset [k]$ be arbitrary real numbers. We divide these variables into two sets, fixing the values of variables in the first set and treating the variables of the second set as free variables. More specifically, let $T\subseteq [m]$ and $\Delta\subseteq 2^{[k]}\setminus\{\emptyset\}$ be arbitrary subsets, and take  $\lambda_{T}$ (shorthand for $\lambda_i$ for $i\in T$) and $\theta_{\Delta}$ (shorthand for $\theta_S$ for $S\in \Delta$) as free variables, and fix the remaining $\lambda_{T^c}$ and $\theta_{\Delta^c}$ as some real numbers.
Then consider the following real valued function  $F_{X_{[k]}}=F_{X_1, \dots, X_k}$ on the free variables and rates
\begin{align}
F_{X_{[k]}}\big(\lambda_{T}, \theta_{\Delta}, R_{[m]}\big)&=\sum_{i=1}^m\lambda_i R_i + \sum_{S\subset [k], S\neq \emptyset}\theta_S H(X_S).
\label{def:F}
\end{align}
By taking maximum over all rates in the capacity region we define
\begin{align}
G_{X_{[k]}}\big(\lambda_{T}, \theta_{\Delta}\big)&=\max_{ R_{[m]}\in\mathcal{R}(X_{[k]})}F_{X_{[k]}}\big(\lambda_{T}, \theta_{\Delta}, R_{[m]}\big).
\label{def:G}
\end{align}
Now, consider the following region in $\mathbb{R}^{|T|+|\Delta|}$ of the values for the free parameters such that $G_{X_{[k]}}$ is not positive:
\begin{align}\Upsilon(X_{[k]})=\big\{&(\lambda_{T}, \theta_{\Delta})|\, G_{X_{[k]}}(\lambda_{T}, \theta_{\Delta})\leq 0\big\}.\label{defupsilon}
\end{align}
The following theorem states that $\Upsilon(X_{[k]})$, which can be understood as the dual of the rate region $\mathcal R(X_{[k]})$, has the tensorization property. 
 
\begin{theorem}\label{thm:main}The function $G_{X_{[k]}}(\lambda_{T}, \theta_{\Delta})$ is additive and the region $\Upsilon(X_{[k]})$ tensorizes. More precisely, for any natural number $n$ we have
\begin{align}
G_{X^n_{[k]}}(\lambda_{T}, \theta_{\Delta})&=n\cdot G_{X_{[k]}}(\lambda_{T}, \theta_{\Delta}), \label{tensequ1}
\end{align}
and
\begin{align}
\Upsilon(X_{[k]}^n)&=\Upsilon(X_{[k]}).
\label{tensequ}
\end{align}
\end{theorem}

\begin{proof} 
Observe that from equation~\eqref{def:F} we have
$$F_{X_{[k]}^n}(\lambda_{T}, \theta_{\Delta}, n R_{[m]})=n F_{X_{[k]}}(\lambda_{T}, \theta_{\Delta}, R_{[m]}).$$
Furthermore, by the additivity of the rate region (equation~\eqref{addveq}) we have $R_{[m]}\in\mathcal{R}(X_{[k]})$ if and only if $n R_{[m]}\in\mathcal{R}(X^n_{[k]})$. This implies equation \eqref{tensequ1}. Equation \eqref{tensequ1} in turn implies~\eqref{tensequ} 
by the definition of $\Upsilon(X_{[k]})$. 
\end{proof}

In the above theorem we prove the additivity of $G_{X_{[k]}}(\lambda_{T}, \theta_{\Delta})$ and the tensorization of $\Upsilon(X_{[k]})$ only in a weak sense, when we consider only i.i.d.\ repetitions of $X_{[k]}$. To prove tensorization in the most general case, i.e., to prove~\eqref{eq:def-tens-region}, we need a stronger version of the additivity of the rate region $\mathcal R(X_{[k]})$ expressed in~\eqref{eq:add-general}.  Indeed assuming that we start with a source coding problem whose rate region satisfies~\eqref{eq:add-general}, the proof of~\eqref{eq:def-tens-region} is obtained by a simple modification of the above argument. However, in this paper  we mostly focus on the tensorization property in its weak sense.  

Observe that Theorem \ref{thm:main} still holds if we more generally replace the entropy function in equation \eqref{def:F} with any other additive function (such as an average cost function).

By the above theorem from any source coding problem we can define a region $\Upsilon(X_{[k]})$ with the tensorization property. Nevertheless, we would like such a region to satisfy the data processing property.

\subsection{Data processing} 
Data processing is another property that we like to prove for $\Upsilon(X_{[k]}).$ That is  for any $$p(y_1,\dots, y_k|x_1, \dots, x_k)=\prod_{i=1}p(y_i|x_i),$$ we would like to have 
\begin{align}
\Upsilon(X_{[k]})\subseteq \Upsilon(Y_{[k]}).\label{eqn:defdata}
\end{align}
The data processing property holds if we can show that $G_{X_{[k]}}$ is decreasing under local stochastic maps, i.e., for any values of $\lambda_{T}$ and $\theta_{\Delta}$ we have
\begin{align}
G_{Y_{[k]}}(\lambda_{T}, \theta_{\Delta})
\leq G_{X_{[k]}}(\lambda_{T}, \theta_{\Delta}).\label{eqn:defdata22}
\end{align}

Data processing does not hold for the dual of any arbitrary source coding problem. Indeed, we should consider an appropriate source coding problem and  an appropriate choice of the fixed parameters $\lambda_{T_1^c}$ and $\theta_{T_2^c}$ for the data processing property to hold.  We have an operational proof of this property when the source coding problem is structured, which we illustrate through concrete examples in the subsequent sections.

\subsection{Connection with initial efficiency} 
Initial efficiency of a rate $R_1$ with respect to a rate $R_2$ is defined as follows  \cite{Erkip,Zhao}.\footnote{
Initial efficiency can be defined more generally in terms of other quantities, e.g.,\ as in capacity per unit cost \cite{Verdu90}.} Let $g(r)$ be the maximum value of $R_1$ when $R_2$ is less than or equal to $r$. That is,
\begin{align}g(r)=\max \{R_1|\, R_{[m]}\in \mathcal R(X_{[k]}), R_2\leq r \}.
\label{def-g-r}
\end{align}
Further assume that $g(0)=0$, meaning that $R_2=0$ implies $R_1=0$. Then $g'(0)$, the derivative of $g(r)$ at $r=0$, is called the initial efficiency of a rate $R_1$ with respect to rate $R_2$. Initial efficiently quantifies how large $R_1$ becomes when we slightly increase $R_2$ from $0$. 

It is not hard to see that the initial efficiency tensorizes by its operational definition when we start with an additive rate region~\cite{CuffInitial, CuffInitial2}. Then the idea of initial efficiently provides a tool to obtain functions with the tensorization property. Here show that this method is a special case of our construction of tensorizing regions, but before that let us clarify the idea of initial efficiency by an example.  

\begin{example}
Let us consider the example of common randomness extraction using one-way communication. Consider two parties who observe i.i.d.\ repetitions of $X$ and $Y$. There is a one-way communication of limited rate $R$ from the first party to the second. Then, the maximum rate of common randomness that can be generated from this source is \cite{AhswedeCsiszar}
$$g(r)=\max_{p(u|x):I(X;U)-I(Y;U)\leq r}I(X;U).$$
By definition $g(0)$ is equal to the G\'acs-K\"orner common information. Assuming that $g(0)=0$, the initial efficiency~\cite{Zhao} is equal to 
$$g'(0)=\lim_{r\searrow 0}\frac{g(r)}{r}=\frac{1}{1-(s^*(X;Y))^2},$$
where $$s^*(X,Y)=\max_{p(u|x)}\frac{I(Y;U)}{I(X;U)}.$$ 
As we discuss later $s^*$ in addition to tensorization satisfies the data processing property as well. 
\end{example}

We now show that initial efficiency can be derived from our construction of tensorizing regions. 
Suppose that the rate region $\mathcal{R}(X_1, \dots, X_m)$ is convex. Then the convexity of $\mathcal{R}(X_1, \dots, X_m)$ implies that $g(r)$ defined in~\eqref{def-g-r} is concave. As a result, from $g(0)=0$ we obtain
\begin{align*}
g'(0)=\lim_{r\searrow 0}\frac{g(r)}{r}=\max_{r\neq 0}\frac{g(r)}{r} =\max_{\stackrel{R_{[m]}\in\mathcal{R}(X_{[k]})}{R_2\neq 0}}\frac{R_1}{R_2}.
\end{align*}
Therefore, $g'(0)$ is equal to the minimum value of $\lambda_2$ such that $R_1-\lambda_2 R_2\leq 0$ for all $R_{[m]}\in \mathcal R(X_{[m]})$.
Then defining $F(\lambda_2, R_{[m]})=R_1-\lambda_2 R_2$, its associated region $\Upsilon$ is equal to $[g'(0), \infty)$. We see that initial efficiency is a special case of our construction of tensorizing regions. 

%**************************************************

\section{Example 1: Lossless source coding with a helper}\label{sec:losslesshleper}
In the problem of source coding with helper, there is a transmitter, a helper and a receiver. The transmitter has access to i.i.d.\ repetitions $X^n$ and the helper has access to $Y^n$ where $(X, Y)$ have a joint distribution $p_{XY}$. The goal of receiver is to recover $X^n$. See Figure~\ref{Fig:LosslessHelper}.

An  $(n, \epsilon, M_1, M_2)$ code for this problem consists of encoder maps $M_1=\mathcal{E}_1(X^n)$ and $M_2=\mathcal{E}_2(Y^n)$, and a decoder map $\hat{X}^n=\mathcal{D}(M_1, M_2)$. The probability of error is equal to $\epsilon=p(\hat{X}^n\neq X^n)$, and the rate pair of this code is $(R_1, R_2)$ where $R_1=\frac{1}{n}\log |\mathcal M_1|$ and $R_2=\frac{1}{n}\log |\mathcal M_2|$. We let $\mathcal R^h(X, Y)$ to be the set of pairs $(R_1, R_2)$ for which there is a sequence of codes $(n, \epsilon_n, M_1, M_2)$ with asymptotic rate $(R_1, R_2)$ such that $\epsilon_n\rightarrow 0$ as $n$ tends to infinity. 

Define 
$$F^h_{X,Y}(\lambda, R_1, R_2)=-\lambda R_1- R_2+\lambda H(X).$$
Observe that $F^h_{X,Y}(\lambda, R_1, R_2)$ has the format of~\eqref{def:F}. Accordingly define 
$$G^h_{X,Y}(\lambda)=\max_{(R_1, R_2)\in \mathcal{R}^h(X,Y)}F^h_{X,Y}(\lambda, R_1, R_2)$$
and 
$$\Upsilon^h(X,Y)=\{\lambda|\, G(\lambda)\leq 0\}.$$
Observe that $(R_1, 0)$ for sufficiently large $R_1$ is in $\mathcal R^h(X, Y)$. Then $\lambda\geq 0$ for any $\lambda\in \Upsilon^h(X,Y)$. 

By Theorem \ref{thm:main}, the set $\Upsilon^h(X,Y)$ tensorizes, i.e. 
\begin{align*}
\Upsilon^h(X^n,Y^n)= \Upsilon^h(X,Y), \qquad \forall n.
\end{align*}
We now show (via an operational proof) that $\Upsilon^h(X,Y)$ also satisfies the data processing property. That is, for all  stochastic maps $p(y'|y)$ and $p(x'|x)$ we have
\begin{align}
\Upsilon^h(X,Y)\subseteq \Upsilon^h(X',Y').
\label{thisfact}
\end{align} 
To prove this it suffices to show that for any $\lambda$ we have
\begin{align}
G^h_{X',Y'}(\lambda)\leq G^h_{X,Y}(\lambda).
\label{eqn:newG1}
\end{align}
By the functional representation lemma~\cite[Appendix B]{elgamal}, any stochastic map can be decomposed as adding some private randomness and application of some function. That is, there are functions $f$ and $g$ such that 
$X'=f(X, A)$ and $Y'=g(X, B)$ where $A$ and $B$ are independent of each other and of $(X, Y)$. Then to show~\eqref{eqn:newG1} we need to prove the followings:
\begin{enumerate}
\item [I.] If $X',Y'$ are functions of $X, Y$ respectively, i.e., if $H(X'|X)=H(Y'|Y)=0$, then $G^h_{X',Y'}(\lambda)\leq G^h_{X,Y}(\lambda)$.
\item[II.]  $G^h_{AX,BY}(\lambda)=G^h_{X,Y}(\lambda)$ if $A$ and $B$ are mutually independent of each other, and of $(X,Y)$.
\end{enumerate}  
Putting the functional representation lemma and the above two cases together, equation~\eqref{eqn:newG1} is implied immediately. In the following we prove the above two claims separately.  
 
\begin{figure}
\begin{center}
\includegraphics[width=3.3in]{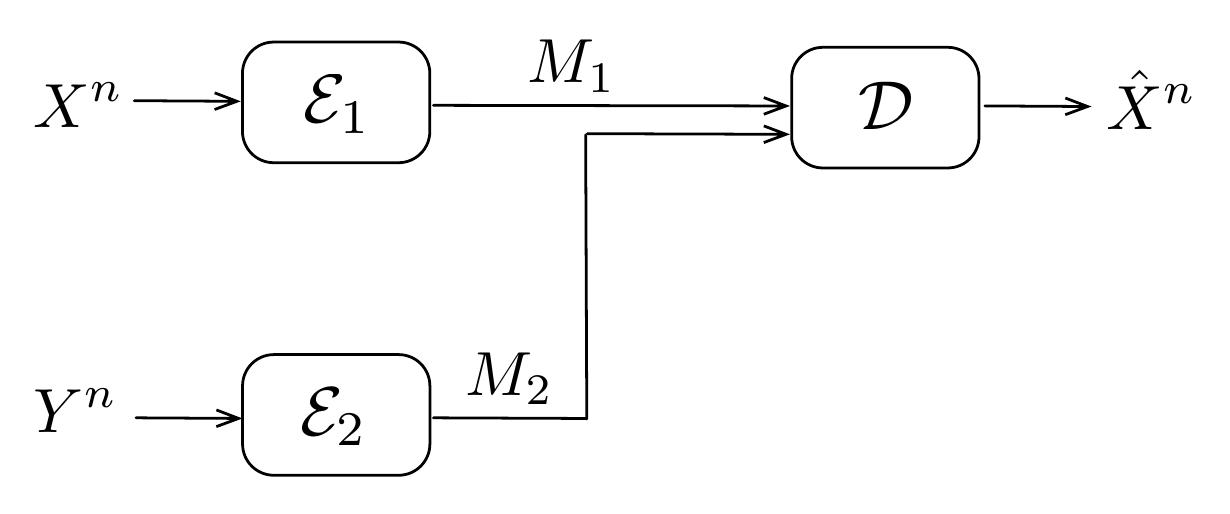}
\caption{Lossless source coding with a helper}
\label{Fig:LosslessHelper}
\end{center}
\end{figure}

\vspace{.15in}
\noindent
\emph{Proof of I.} We need to show that for any $\lambda$
\begin{align*}\max_{(R'_1, R'_2)\in\mathcal{R}^h(X',Y')}-\lambda R'_1- R'_2+\lambda H(X')
\leq \max_{(R_1, R_2)\in\mathcal{R}^h(X,Y)}-\lambda R_1- R_2+\lambda H(X).\end{align*}
Using the fact that $H(X)=H(XX')=H(X')+H(X|X')$, it suffices to show that if $(R'_1, R'_2)\in\mathcal{R}^h(X',Y')$, then $(R_1, R_2)=(R'_1+H(X|X'), R'_2)\in\mathcal{R}^h(X,Y)$.
To show this, fix a code for the source $(X',Y')$ with rate pair of $(R'_1, R'_2)$. Now consider the following protocol for the source $(X, Y)$: the transmitter and helper compute $X', Y'$ from $X, Y$ respectively, and then use the above code to send $X'$ to the receiver. Then using the Slepian-Wolf theorem, the transmitter by sending $H(X|X')$ extra bits (on average) sends $X$ to the receiver. In this protocol the helper sends information at rate $R_2=R'_2$ and the transmitter sends information at rate $R_1=R'_1+H(X|X')$.

\vspace{.15in}
\noindent
\emph{Proof of II.} From the definition of the source coding problem it is clear that $G^h_{AX,BY}(\lambda)=G^h_{AX,Y}(\lambda)$ since $B$ has the role of private randomness of the helper. It remains to show that 
$G^h_{AX,Y}(\lambda)=G^h_{X,Y}(\lambda)$. Since $X$ is a function of $(A, X)$, using part I we have
$$G^h_{X,Y}(\lambda)\leq G^h_{AX,Y}(\lambda).$$
Thus, we need  to show that $G^h_{X,Y}(\lambda)\geq G^h_{AX,Y}(\lambda)$, or equivalently
\begin{align*}\max_{(R_1, R_2)\in\mathcal{R}^h(X,Y)}-\lambda R_1- R_2+\lambda H(X)\geq 
\max_{(R_1, R_2)\in\mathcal{R}^h(AX,Y)}-\lambda R_1- R_2+\lambda H(AX).\end{align*}
To prove this we show that for any $(R_1, R_2)\in\mathcal{R}^h(AX,Y)$, we have $(R_1-H(A), R_2)\in\mathcal{R}^h(X,Y)$. To show this, we again use the Slepian-Wolf theorem. 

Fix $(R_1, R_2)\in\mathcal{R}^h(AX,Y)$ and a sequence of codes $(n, \epsilon_n, M_1, M_2)$ achieving this point. Since $M_2$ is generated from $Y^n$, it is independent of $A^n$. Then using the Fano inequality we have
\begin{align*}
H(M_1|M_2A^n)&=H(M_1|M_2)-I(M_1;A^n|M_2)
\\&= H(M_1|M_2)-I(M_1M_2;A^n)
\\&\leq H(M_1)-H(A^n)+o(n)
\\&= n\big(R_1-H(A)+o(1)\big),
\end{align*}
where in the third line we use the fact that $A^n$ can be recovered from $(M_1, M_2)$ with probability at least $1-\epsilon_n$. Next, following similar ideas we have
\begin{align}H(X^n|M_2A^n)&=H(X^nM_1|M_2A^n)\nonumber
\\&= H(M_1|M_2A^n)+H(X^n|M_1M_2A^n)\nonumber
\\&\leq H(M_1|M_2A^n)+o(n)\nonumber
\\&\leq n\big(R_1-H(A)+o(1)\big),
\label{thfa2}
\end{align}
where in the last line we use the previous inequality. 

We now construct a protocol that shows $(R_1-H(A), R_2)\in\mathcal{R}^h(X,Y)$. Think of $A$ as shared randomness between the transmitter and the receiver. Note that shared randomness does not change the rate region $\mathcal{R}^h(X,Y)$. In the new protocol the helper uses the same encoding map to create $M_2$ from $Y^n$. Then the receiver has $A^n$ in hand and gets $M_2$ from the helper. 
%Now, the uncertainty of the receiver about $X^n$ is equal to $H(X^n|M_2A^n)\leq n\big(R_1-H(A)+g_2(\epsilon)\big)$ by equation \eqref{thfa2}.
Then by the Slepian-Wolf theorem, if we consider $N$ i.i.d.\ repetitions of this code, the transmitter needs to send only $H(X^n|M_2A^n)+o(n)$ bits on average to convey $X^n$ to the receiver. 
%In other words, we take some large $N$ and consider $nN$ repetitions of $(X,Y,A)$. Then we consider $X^{nN}$ and then apply the Slepian-Wolf decoding at the transmitter on the whole $X^{nN}$ sequence, viewing it as $N$ i.i.d.\ repetitions of $X^n$. 
In this protocol the rate of communication from the helper is $R_2$ and the rate of 
communication from transmitter is $\frac{1}{n}H(X^n|M_2A^n)+o(1)$ which using~\eqref{thfa2} is at most $R_1-H(A)+o(1)$. Then $(R_1-H(A), R_2)\in \mathcal R^h(X, Y)$.

\vspace{.15in}
By the above discussion $\Upsilon^h(X,Y)$ satisfies the tensorization and data processing properties. Note that for proving these properties, we did not use the characterization of the capacity region $\mathcal R^h(X, Y)$; we proved these properties via operational arguments and used only the Slepian-Wolf theorem.  Nevertheless, we may use the characterization of $\mathcal R^h(X, Y)$ to compute $\Upsilon^h(X, Y)$.

From \cite[Theorem 10.2]{elgamal} the capacity region $\mathcal{R}^h(X,Y)$ is equal to the set of pairs $(R_1, R_2)$ satisfying
\begin{align*}R_1&\geq H(X|U),
\qquad R_2\geq I(Y;U),
\end{align*}
for some conditional distribution $p(u|y)$. 
Then for non-negative values of $\lambda$ we have \begin{align}G^h_{X,Y}(\lambda)= \max_{U-Y-X}\lambda I(X;U)- I(Y;U).\label{defGh}
\end{align}
Therefore, $\lambda\in \Upsilon(X, Y)$ if and only if $\lambda I(X;U)- I(Y;U)\leq 0$ for all $p(u|y)$. Equivalently, $\lambda\in \Upsilon^h(X,Y)$ iff 
$$\frac{1}{\lambda}\geq \max_{U-Y-X}\frac{I(X;U)}{I(Y;U)}=s^*(Y, X).$$
Therefore, our discussion above provides a proof for the fact that $s^*(Y, X)$ tensorizes and satisfies the data processing inequality.
 
By the above discussion $s^*(Y, X)$ is the initial efficiency of the one-helper source coding problem: let $h(R_2)$ be the minimum value of $R_1$ for a given $R_2$. Then $h(0)=H(X)$. Let $g(R_2)=h(0)-h(R_2)$. Then  
$$s^*(Y, X)=\max_{\stackrel{R_{[2]}\in \mathcal R^h(X, Y)}{R_2\neq 0}}\frac{g(R_2)}{R_2}.$$
\color{black}

%********************************************

\section{Example 2: One side-information source problem}\label{sec:ex2}
%page 375 of the pdf.

The one side-information source problem \cite[Problem 16.6 (c)]{csiszarbook} is a generalization of the problem considered in Section \ref{sec:losslesshleper}. Here there are $k$ transmitters, one helper and $k$ receivers. Transmitter $i$, $1\leq i\leq k$, observes i.i.d.\ repetitions $X_i^n$ and the helper observes i.i.d.\ repetitions $X_{k+1}^n$. The $i$-th transmitter sends information at rate $R_i$ to receiver $i$, and helper broadcasts information to all receivers at rate $R_{k+1}$. The goal of the $i$-th receiver is to recover $X_i^n$. See Figure~\ref{fig:OneSideInf}. We denote the set of achievable rate tuples $(R_1, \dots, R_{k+1})$ for this problem by $\mathcal R^s(X_1, \dots, X_{k+1})$.

To obtain a dual for this rate region let us define
 \begin{align*}
 F^s_{X_{[{k+1}]}}(\lambda_{[k]}, R_{[{k+1}]})=-R_{k+1}-\sum_{i=1}^{k}\lambda_iR_i +\sum_{i=1}^{k}\lambda_iH(X_i).
\end{align*}
Then let
\begin{align*}
G^s_{X_{[{k+1}]}}(\lambda_{[k]})&=\max_{R_{[{k+1}]}\in\mathcal{R}^s(X_{[{k+1}]})}F^s_{X_{[{k+1}]}}(\lambda_{[k]}, R_{[{k+1}]}),
\end{align*}
and
\begin{align*}
\Upsilon^s(X_{[{k+1}]})&=\{\lambda_{[k]}|\, G^s_{X_{[{k+1}]}}(\lambda_{[k]})\leq 0\}.
\end{align*}
Again for sufficiently large $R_1, \dots, R_{k}$ we have $(R_1, \dots, R_{k}, 0)\in \mathcal R^s(X_{[{k+1}]})$. Then for any $\lambda_{[k]}\in \Upsilon^s(X_{[{k+1}]})$ we have $\lambda_i\geq 0$.

\begin{figure}
\begin{center}
\includegraphics[width=3.3in]{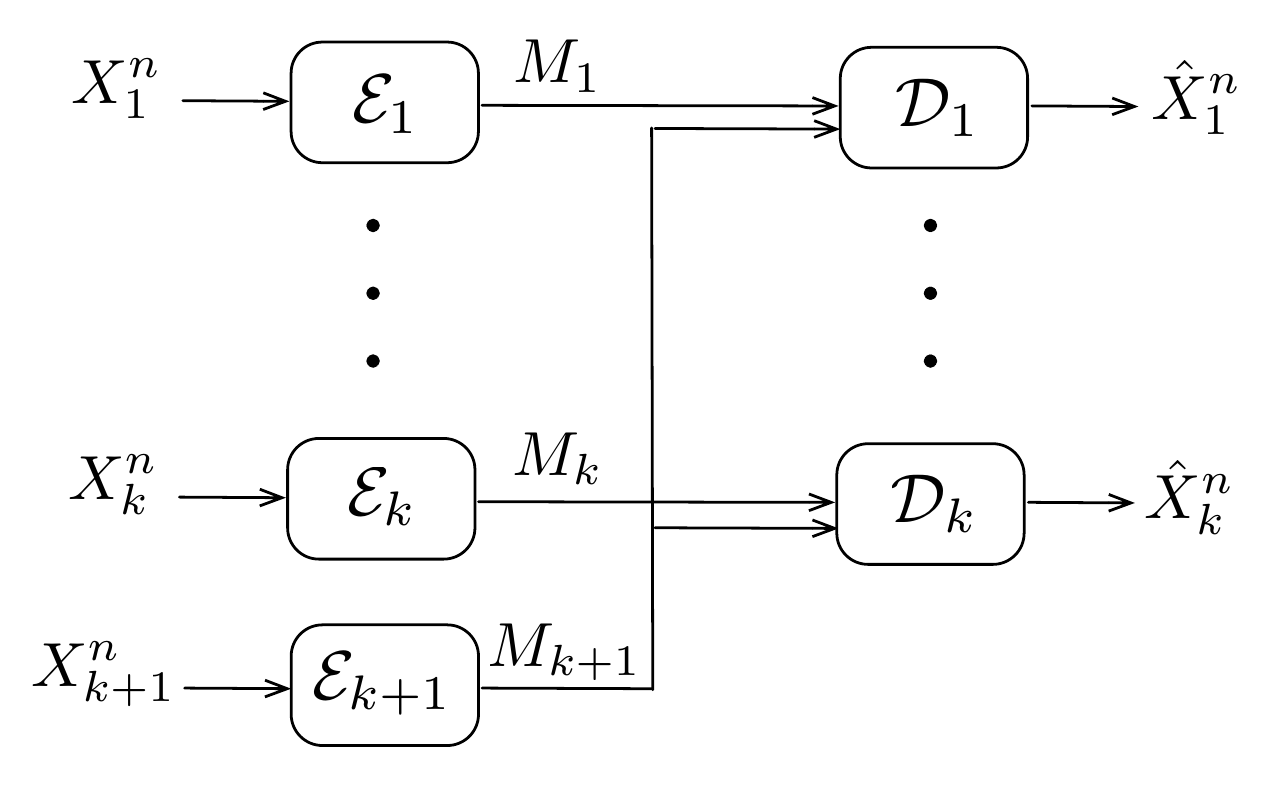}
\caption{One side-information source problem}\label{fig:OneSideInf}
\end{center}
\end{figure}

By Theorem~\ref{thm:main} the function $G^s_{X_{[{k+1}]}}(\lambda_{[k]})$ is additive and the set $\Upsilon^s(X_{[{k+1}]})$ satisfies tensorization. We claim that $\Upsilon^s(X_{[{k+1}]})$ also satisfies the data processing property. To prove this claim it suffices to 
show that for any $p(x'_i|x_i)$ we have
$$G^s_{X'_{[{k+1}]}}(\lambda_{[k]})\leq G^s_{X_{[{k+1}]}}(\lambda_{[k]}).$$ 
The proof of this inequality is completely similar to the proof of~\eqref{eqn:newG1} given in the previous section and we do not repeat it in full details here. Briefly speaking, as before we first use the functional representation lemma to break the proof in two parts. We first consider 
the case where $X'_i$ is a function of $X_i$; here we argue that  it suffices to show that if $R'_{[{k+1}]}\in\mathcal{R}^s(X'_{[{k+1}]})$, then 
$$\big(R'_1+H(X_1|X'_1), \dots, R'_{k}+H(X_{k}|X'_{k}), R'_{k+1}\big)\in\mathcal{R}^s(X_{[{k+1}]}).$$ 
This follows again from the Slepian-Wolf theorem. Next, we show that 
$G^s_{A_{[{k+1}]}X_{[{k+1}]}}(\lambda_{[k]})=G^s_{X_{[{k+1}]}}(\lambda_{[k]})$ when $A_{1}, \dots, A_{k+1}$ are independent of each other of of $X_{[{k+1}]}$. For this we show that 
 if $R^s_{[{k+1}]}\in\mathcal{R}^s(A_1X_1,\dots, A_{k}X_{k}, X_{k+1})$, then 
 $$(R_1-H(A_1), \dots, R_{k}-H(A_{k}), R_{k+1})\in\mathcal{R}^s(X_{[{k+1}]}).$$ 
 This follows again from thinking of $A_{[k]}$ as shared randomness among the parties and using the Fano inequality and Slepian-Wolf theorem.

Now we have region $\Upsilon^s(X_{[{k+1}]})$ that tensorizes and satisfies data processing. Using~\cite[Problem 16.6 (c)]{csiszarbook}, the capacity region $\mathcal R^s(X_{[{k+1}]})$ of this problem is given by 
\begin{align}
R_{k+1}&\geq I(U;X_{k+1}),\label{eq:gw-1f}\\
R_i&\geq H(X_i|U),\quad \forall i\in[k].\label{eq:gw-2f}
\end{align}
for some $U-X_{k+1}-X_{[k]}$. Therefore, for non-negative tuples $\lambda_{[k]}$, we have
\begin{align}
G^s_{X_{[{k+1}]}}(\lambda_{[k]})=\max_{U-X_{k+1}-X_{[k]}}-I(X_{k+1};U)+\sum_{i=1}^{k}\lambda_i I(X_i;U).\label{defGinV6}
\end{align}
As a result, $\lambda_{[k]}\in \Upsilon^s(X_{[{k+1}]})$ iff 
$$\sum_{i=1}^{k}\lambda_i I(X_i;U)\leq I(X_{k+1};U),$$
for every $U-X_{k+1}-X_{[k]}$. The following theorem summarizes the above findings.

\begin{theorem}\label{defreg1} 
For any distribution $p_{X_{[{k+1}]}}$ let $\Upsilon^s(X_{[{k+1}]})$ be the set of all non-negative $\lambda_{[k]}$ such that 
$$\sum_{i=1}^{k}\lambda_i I(X_i;U)\leq I(X_{k+1};U),$$
for all $p(u|x_{k+1})$. Then $\Upsilon^s(X_{[{k+1}]})$ satisfies the data processing inequality and tensorization.
\end{theorem}

The region $\Upsilon^s(X_{[{k+1}]})$ is non-empty; by data processing inequality if $U-X_{k+1}-X_{[k]}$ forms a Markov chain, we have $I(X_i;U)\leq I(X_{k+1};U)$. Then $\Upsilon^s(X_{[{k+1}]})$ includes any $\lambda_{[k]}$ satisfying 
$0\leq \lambda_i$ and $\sum_{i=1}^{k}\lambda_i\leq 1$.

\begin{example}
Consider the special case where $k=2$ and $X_3=(X_1, X_2)$. In this case $\Upsilon^s(X_{[3]})$ is equivalent to the following region:
$$ \fR(X_1, X_2)=\big\{(\lambda_1, \lambda_2)\in \mathbb{R}_{+}^2|\, \lambda_1 I(X_1;U)+\lambda_2 I(X_2;U)\leq I(X_1X_2;U)\big\}.$$
Then $\fR(X_1, X_2)$ satisfies tensorization and data processing properties. 
\end{example}

Observe that in the special case of  $k=2$ and $X_3=(X_1, X_2)$, the rate region given in equations \eqref{eq:gw-1f} and \eqref{eq:gw-2f} reduces to that of the Gray-Wyner rate region~\cite{GrayWyner}. Then $\fR(X_1, X_2)$ can be understood as the dual of the Gray-Wyner region.

By the following theorem of Nair~\cite{Nair} gives another characterization of $\fR(X_1, X_2)$ defined above.

\begin{theorem}[\cite{Nair}] $(\lambda_1, \lambda_2)\in \fR(X_1, X_2)$ if and only if for every pair of functions $f_{X_1}:\mathcal{X}_1\rightarrow \mathbb R$ and $g_{X_2}:\mathcal {X}_2\rightarrow \mathbb R$ we have
\begin{align}\label{eq:fg-norm-2}
\E[f_{X_1}g_{X_2}]\leq \|f_{X_1}\|_{\frac{1}{\lambda_1}}\|g_{X_2}\|_{\frac{1}{\lambda_2}},
\end{align}
where the Schatten norms are defined by $\|f_{X_1}\|_{\frac{1}{\lambda_1}}=\E\big[|f_{X_1}|^{1/\lambda_1}\big]^{\lambda_1}$ and similarly for $\|g_{X_2}\|_{\frac{1}{\lambda_2}}$.
\label{thm:nair}
\end{theorem}

The set of pairs $(\lambda_1, \lambda_2)$ satisfying~\eqref{eq:fg-norm-2} is the hypercontractivity ribbon defined in~\cite{AhlswedeGacs}. Hypercontractivity ribbon is known to satisfy the data processing and tensorization. The above theorem gives an alternative characterization of the hypercontractivity ribbon.

Another interesting property of hypercontractivity ribbon is that it characterizes $s^*(X, Y)$ as follows: 
\begin{align}\label{eq:hc-ribbon-s-star}
s^*(X, Y)=\inf_{(\lambda_1, \lambda_2)\in\fR(X, Y)}\frac{1-\lambda_1}{\lambda_2}.
\end{align}
For a proof of this equation see~\cite{MC-HC}.

\begin{example}[Multipartite hypercontractivity ribbon] In Theorem~\ref{defreg1} assume that ($k$ is arbitrary and) $X_{k+1}=(X_1, \dots, X_{k})$. Then $\Upsilon_1(X_{[{k+1}]})$ reduces to 
$$ \fR(X_{[k]})=\big\{\lambda_{[k]}\in \mathbb{R}_{+}^{k}\big|\, \sum_{i=1}^{k}\lambda_i I(X_i;U)\leq I(X_{[k]};U)\big\}.$$
As a result, $\fR(X_{[k]})$ satisfies data processing and tensorization.
\label{ex:M-HC-ribbon} 
\end{example}

Letting $U=X_i$ we observe that if $\lambda_{[k]}\in \fR(X_{[k]})$ then $\lambda_i\leq 1$. Therefore, 
$$\fR(X_{[k]})\subseteq [0,1]^{k}.$$ 
Furthermore, since $\fR(X_{[k]})$ is a special case of regions of the form $\Upsilon^s$, it includes any $\lambda_{[k]}$ satisfying 
$0\leq \lambda_i$ and $\sum_{i=1}^{k}\lambda_i\leq 1$, as argued above. 

The multipartite hypercontractivity ribbon is equal to  $[0,1]^{k}$ if and only if $X_i$ are mutually independent. To prove this note that if $(1,1,\dots, 1)\in  \fR(X_{[k]})$ then by setting $U=X_{[k]}$ we find that $\sum_{i=1}^{k} H(X_i)\leq H(X_{[k]})$. Then by the subadditivity inequality of entropy, $X_i$'s are mutually independent. On the other hand, for mutually  independent variables $X_i$ we have  $\sum_{i=1}^{k} H(X_i)= H(X_{[k]})$  and $\sum_{i=1}^{k} H(X_i|U)\leq H(X_{[k]}|U)$. This shows that $(1,1,\dots, 1)\in  \fR(X_{[k]})$. It is straightforward to generalize Theorem \ref{thm:nair} of \cite{Nair} to show that the multipartite region $\fR(X_{[k]})$ has a characterization in terms of Schatten norms.

%*****************************************************

\section{Example 3: Fork network with side information}\label{sec:fork}
The fork network with side information is another generalization of the problem we studied in Section \ref{sec:losslesshleper} (see \cite[Problem 16.31]{csiszarbook}, \cite[Theorem 10.4]{elgamal}). The difference of this problem with the one considered in Section \ref{sec:ex2} is that there is only one decoder who needs to recover $X_{[k]}$. The problem is depicted in Figure~\ref{fig:ForkNet}. We denote the capacity region of this problem by $\mathcal R^f(X_1, \dots, X_{k+1})$.%page 395 of the pdf.

As in Section~\ref{sec:ex2}, define 
\begin{align*}
F^f_{X_{[k+1]}}(\lambda_{[k]}, R_{[k+1]})&=-R_{k+1}-\sum_{i=1}^{k}\lambda_iR_i +\sum_{i=1}^{k}\lambda_iH(X_i),
\\G^f_{X_{[{k+1}]}}(\lambda_{[k]})&=\max_{R_{[{k+1}]}\in\mathcal{R}^f(X_{[{k+1}]})}F^f_{X_{[k]}}(\lambda_{[k]}, R_{[{k+1}]}),
\end{align*}
and
\begin{align*}
\Upsilon^f(X_{[{k+1}]})&=\big\{\lambda_{[k]}\big|\, G^f_{X_{[{k+1}]}}(\lambda_{[k]})\leq 0\big\}.
\end{align*}
As in the previous two sections, $\Upsilon^f(X_{[{k+1}]})$ may only contain  non-negative tuples $\lambda_{[k]}$. Again Theorem~\ref{thm:main} implies that $G^f_{X_{[{k+1}]}}(\lambda_{[k]})$ is additive and the set $\Upsilon^f(X_{[{k+1}]})$ tensorizes. 

We claim that $\Upsilon^f(X_{[{k+1}]})$ also satisfies the data processing property. To show this, we prove that for any $p(x_i'|x_i)$ we have
$$G^f_{X'_{[{k+1}]}}(\lambda_{[k]})\leq G^f_{X_{[{k+1}]}}(\lambda_{[k]}).$$ 
Again we split the proof in two parts. 
When $X'_i$ is a function of $X_i$, the proof is identical to the one given in the Section~\ref{sec:ex2}. It remains to show that $G^f_{A_{[{k+1}]}X_{[{k+1}]}}(\lambda_{[k]})=G^f_{X_{[{k+1}]}}(\lambda_{[k]})$ when $A_1, \dots, A_{k+1}$ are independent of each other and of $X_{[{k+1}]}$.
For this we need to show that
 if 
 $$(R_1, \dots, R_{k+1})\in\mathcal{R}^f(A_1X_1,\dots, A_{k}X_{k}, X_{k+1}),$$ 
 then 
 $$(R_1-H(A_1),  \dots, R_{k}-H(A_{k}), R_{k+1})\in\mathcal{R}^f(X_1,\dots,  X_{k+1}).$$ 
 To prove this last claim, we follow similar ideas as before. We start with sequence of 
 $(n,\epsilon_n, M_1, \dots, M_{k+1})$ codes with asymptotic rate tuple $(R_1, \dots, R_{k+1})\in\mathcal{R}^f(A_1X_1,\dots, A_{k}X_{k}, X_{k+1})$. 
 Take a non-empty subset $S\subseteq[k]$. Letting $S^c=[k]-S$, we have 
\begin{align} 
H(X^n_S|\,M_{k+1}A_{[k]}^nX^n_{S^c})&=H(X^n_SM_S|\, M_{k+1}A_{[k]}^nX^n_{S^c}M_{S^c})\nonumber
\\&=
H(M_S|\,M_{k+1}A_{[k]}^nX^n_{S^c}M_{S^c})+H(X^n_S|\, A_{[k]}^nM_{[{k+1}]}X^n_{S^c})\nonumber
\\&\leq H(M_S|\, M_{k+1}A_{[k]}^nX^n_{S^c}M_{S^c})+o(n)\label{eqn:Fanoz1}
\\&=H(M_S|\, M_{k+1}A_{S^c}^nX^n_{S^c}M_{S^c})-I(A_{S}^n;M_S|\, M_{k+1}A_{S^c}^nX^n_{S^c}M_{S^c})+o(n)\nonumber
\\&\leq H(M_S)-H(A_{S}^n|M_{k+1}A_{S^c}^nX^n_{S^c}M_{S^c})+H(A_{S}^n|A_{S^c}^nX^n_{S^c}M_{[{k+1}]})+o(n)\nonumber
\\&= H(M_S)-H(A_{S}^n)+H(A_{S}^n|A_{S^c}^nX^n_{S^c}M_{[{k+1}]})+o(n)\label{eqn;useinde}
\\&= H(M_S)-H(A_{S}^n)+o(n)\label{eqn:Fanoz2}
\\&\leq \sum_{i\in S}\big(H(M_i)-nH(A_i)\big)+o(n)\label{eqn:indeus23}
\\&=n \Big(\sum_{i\in S}\big(R_i-H(A_i)\big)+o(1)\Big).\label{eqn:finalMSMK}
\end{align}
Here equations~\eqref{eqn:Fanoz1} and~\eqref{eqn:Fanoz2} follow from Fano's inequality; equation \eqref{eqn;useinde} follows from the fact that $A_{S}^n$ is independent of $A_{S^c}^nX^n_{[k+1]}$ and then of $M_{k+1}A_{S^c}^nX^n_{S^c}M_{S^c}$; finally equation \eqref{eqn:indeus23} uses the fact that $A_i$'s are mutually independent.

\begin{figure}
\begin{center}
\includegraphics[width=3.3in]{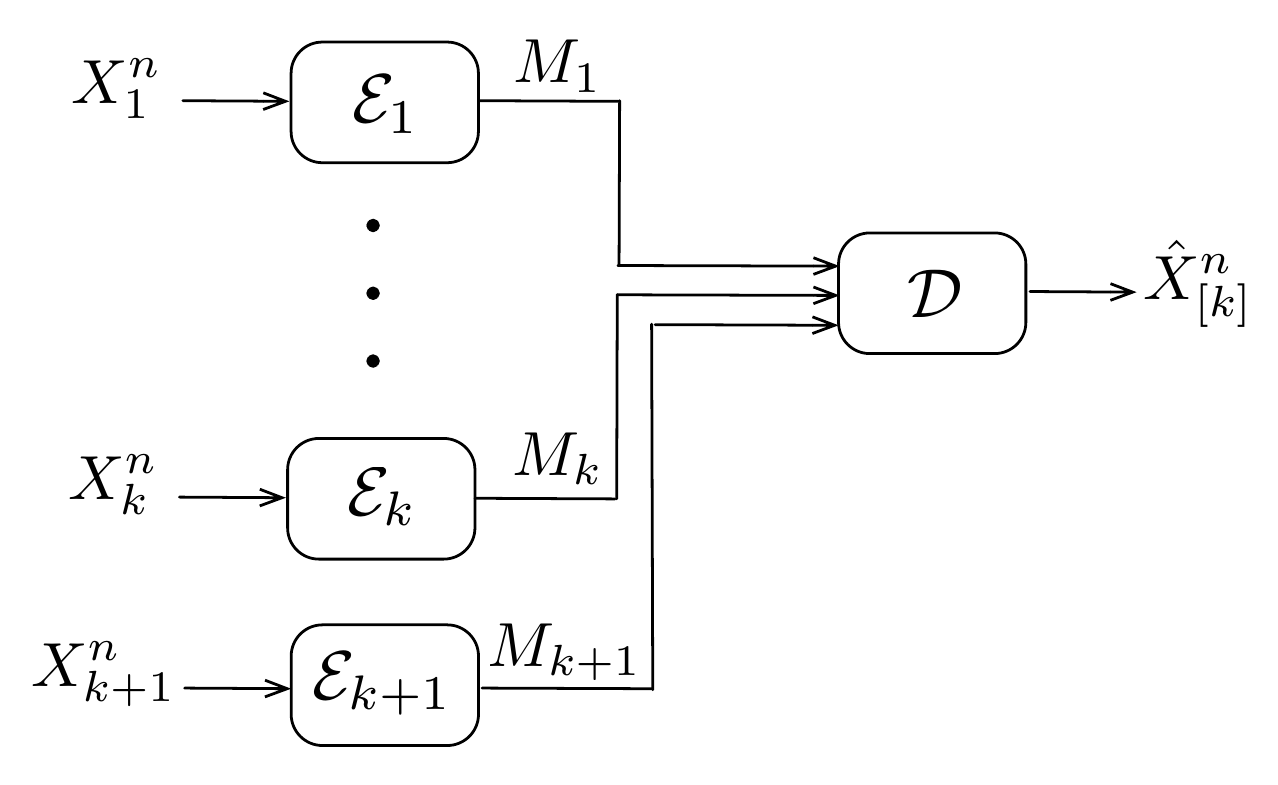}
\caption{Fork network with side information}\label{fig:ForkNet}
\end{center}

\end{figure}

Now we construct a code for inputs $X_{[k+1]}$. We think of $A_{[k]}^n$ as shared randomness given to all the parties. We assume that the encoder $k+1$ creates side information $M_{k+1}$ and sends it to the receiver as before. Then the receiver has side information $M_{k+1}A_{[k]}^n$ and wants to decode $X_{[k]}^n$. To this end, we use the Slepain-Wolf theorem which states that the recovery of $X_{[k]}^n$ is possible if the $i$-th transmitter, for $1\leq i\leq k$, sends information at rate $R'_i$ assuming that for every subset $S\subseteq[k]$ we have
$$\sum_{i\in S}{R'}_i\geq H(X_S^n|\, X_{S^c}^nM_{k+1}A_{[k]}^n),$$
where $S^c=[k]-S$. However, from \eqref{eqn:finalMSMK} we have
$$H(X^n_S|M_{k+1}A_{[k]}^nX^n_{S^c})\leq  n \Big(\sum_{i\in S}\big(R_i-H(A_i)\big)+o(1)\Big).$$
Therefore, if we set ${R'}_i=n(R_i-H(A_i)+o(1))$, the necessary conditions of the Slepain-Wolf theorem with side information at the decoder are satisfied. Thus, we can transmit $N$ repetitions of $X_i^n$ at the average rate of $n(R_i-H(A_i)+o(1))$. This shows that 
$$(R_1-H(A_1), \dots, R_{k}-H(A_{k}), R_{k+1})\in\mathcal{R}^f(X_1,\dots, X_{k+1}).$$
\vspace{0.5cm}

The above discussion implies that $\Upsilon^f(X_{[{k+1}]})$ satisfies data processing and tensorization. 

According to~\cite[Theorem 10.4]{elgamal}, the capacity region $\mathcal R^f(X_{[{k+1}]})$ consists of tuples $R_{[{k+1}]}$ such that  
\begin{align}
R_{k+1}&\geq I(U;X_{k+1}),\\
\sum_{i\in S}R_i&\geq H(X_S|UX_{S^c}),\quad \forall S\subset [k],\label{eq:gw-3f}
\end{align}
for some $U-X_{k+1}-X_{[k]}$. 

Let us consider the special case $k=2$. Then, the rate region is described by
\begin{align*}
R_3&\geq I(U;X_{3}),\\
R_1&\geq H(X_1|UX_{2}),\\
R_2&\geq H(X_2|UX_{1}),\\
R_1+R_2&\geq H(X_1X_2|U).
\end{align*}
The corner points of this region are 
$$(R_1, R_2, R_3)=\big(H(X_1|X_2U), H(X_2|U), I(U;X_3)\big),$$
and 
$$(R_1, R_2, R_3)=\big(H(X_1|U), H(X_2|X_1U), I(U;X_3)\big).$$
Since $G^f_{X_{[k+1]}}(\lambda_{[k]})$ involves maximization of a linear function, its maximum occurs at one of these corner points. Then one can verify that for non-negative values of $\lambda_1$ and $\lambda_2$ we have
$$G^f_{X_{[3]}}(\lambda_{1}, \lambda_2)=\max_{U-X_3-X_1X_2}-I(X_3;U)+\lambda_1 I(X_1;U)+\lambda_2 I(X_2;U)+\max\{\lambda_1, \lambda_2\}I(X_1;X_2|U).$$
Hence, we have the following theorem.

\begin{theorem} 
The following region satisfies data processing and tensorization:
\begin{align*}
\Upsilon^f(X_1, X_2, X_3)=\big\{(\lambda_1, \lambda_2)\in \mathbb{R}_{+}^2|\, &\lambda_1 I(X_1;U)+\lambda_2 I(X_2;U)\\&+\max\{\lambda_1, \lambda_2\}I(X_1;X_2|U)\leq I(X_3;U),\quad \forall p(u|x_3)\big\}.
\end{align*}
\end{theorem}

Note that the above region differs from the hypercontractivity ribbon as it includes the term $\max\{\lambda_1, \lambda_2\}I(X_1;X_2|U)$.

By setting $U$ to be a constant random variable, we observe that $\Upsilon^f(X_1, X_2, X_3)=\{(0,0)\}$ if $I(X_1;X_2)>0$. Therefore, to get a non-trivial region one must have $I(X_1;X_2)=0$. 
Assuming this and using the expansion $I(X_1;X_2|U)-I(X_1;X_2)=-I(X_1;U)-I(X_2;U)+I(X_1X_2;U)$, we observe that $\Upsilon^f(X_1, X_2, X_3)$ has the following alternative characterization (when $I(X_1; X_2)=0$):
\begin{align}
\Upsilon^f(X_1, X_2, X_3)=\big\{(\lambda_1, \lambda_2)\in \mathbb{R}_{+}^2|\, &\min\{0, \lambda_1-\lambda_2\} I(X_1;U)+\min\{0, \lambda_2-\lambda_1\} I(X_2;U)\nonumber\\&+\max\{\lambda_1, \lambda_2\}I(X_1X_2;U)\leq I(X_3;U),\quad \forall p(u|x_3)\big\}.\label{eqn:alternativeUpsilon3}
\end{align}
The above expression allows for an explicit characterization of the set of pairs $(\lambda, \lambda)\in \Upsilon_2(X_1, X_2, X_3)$. Indeed, $(\lambda, \lambda)$ is in $\Upsilon(X_1, X_2, X_3)$ if and only if
$$\frac{1}{\lambda}\geq \max_{U-X_3-X_1X_2}\frac{I(X_1X_2;U)}{I(X_3;U)}=s^*(X_3;X_1X_2).$$

%********************************************

\section{Conditional tensorization}\label{sec:cond:HC}
Consider the source coding problem of Section \ref{sec:ex2}. Let us provide all of the parties (encoders and decoders) with i.i.d.\ repetitions of some random variable $Z$, which is jointly distributed with $X_{[{k+1}]}$.  This is similar to the idea of Coded Time Sharing \cite[Sec. 4.5.3]{elgamal}. Then one can see that the capacity region $\mathcal R^c(X_1,\dots, X_{k+1}, Z)$ for this problem is equal to the one given in equations~\eqref{eq:gw-1f} and~\eqref{eq:gw-2f}, except that everything gets conditioned on $Z$:
\begin{align}
R_{k+1}&\geq I(U;X_{k+1}|Z),\label{eq:gw-1zf}\\
R_i&\geq H(X_i|UZ),\quad \forall i\in[k],\label{eq:gw-2zf}
\end{align}
for some $U-X_{k+1}Z-X_{[k]}$. This region results in the following region $\Upsilon^c(X_1, \dots, X_{k+1}, Z)$ consisting of all non-negative $\lambda_i$ such that 
$$\sum_{i=1}^{k}\lambda_i I(X_i;U|Z)\leq I(X_{k+1};U|Z),$$
for all $p(u|zx_{k+1})$.

Let us consider the special case of $k=2$, $X_3=(X_1, X_2)$:
\begin{theorem}[Conditional bipartite hypercontractivity ribbon] 
Let
$$ \fR(X_1, X_2|Z)=\{(\lambda_1, \lambda_2)\in \mathbb{R}_{+}^2|\, \lambda_1 I(X_1;U|Z)+\lambda_2 I(X_2;U|Z)\leq I(X_1X_2;U|Z), \quad \forall U\}.$$
Then we have
\begin{itemize}
\item Tensorization: 
$\fR(X_1, X_2|Z)=\fR(X^n_1, X^n_2|Z^n)$ if $(X^n_1, X^n_2,Z^n)$ is $n$ i.i.d.\ repetitions of $(X_1, X_2, Z)$.
\item Data processing: 
$\fR(X_1, X_2|Z)\subseteq \fR(X'_1, X'_2|Z)$ for any $p(x'_1|x_1z)$ and $p(x'_2|x_2z)$. 
\end{itemize}
\end{theorem}
The above properties of the conditional hypercontractivity ribbon can be operationally proved as before. Alternatively we have the following characterization of the conditional hypercontractivity ribbon from which the above theorem is implied. 

\begin{lemma}\label{lemma9} We have
\begin{align}\fR(X_1, X_2|Z)=\bigcap_{z:p(z)>0}\fR(X_1, X_2|Z=z).\label{eqn:eq-form-HC}\end{align}
\end{lemma}

\begin{proof}
It suffices to show that 
\begin{align}
\lambda_1 I(X_1;U|Z)+\lambda_2 I(X_2;U|Z)\leq I(X_1X_2;U|Z), \qquad \forall U
\label{eqnZ1}\end{align}
if and only if 
\begin{align}\lambda_1 I(X_1;U|Z=z)+\lambda_2 I(X_2;U|Z=z)\leq I(X_1X_2;U|Z=z), \qquad \forall U
\label{eqnZ2}\end{align}
for all $z$ with $p(z)>0$. Clearly, equation~\eqref{eqnZ2} implies \eqref{eqnZ1}. 
To see the converse, given any arbitrary $z^*$, observe that we can choose $U$ to be a constant if $z\neq  z^*$. 
\end{proof}

One can similarly define conditional $s^*(X_1,X_2|Z)$ either using~\eqref{eq:hc-ribbon-s-star} as
$$s^*(X_1, X_2|Z)=\inf_{(\lambda_1, \lambda_2)\in\fR(X_1, X_2|Z)}\frac{1-\lambda_1}{\lambda_2},$$
or directly from the source coding problem of Section~\ref{sec:losslesshleper} as
$$s^*(X_1, X_2|Z)=\max_{U-ZX_1-X_2}\frac{I(X_2, U|Z)}{I(X_1, U|Z)}=\max_{z~\text{where}~p(z)>0}s^*(X_1, X_2|Z=z).$$
These two definition coincide as can be verified using their equivalency in the unconditional case. 
Moreover, they match with the definition of $s^*_{Z}(X_1Z, X_2Z)$ given in~\cite{CuffInitial}. In Appendix \ref{conditionalRhoSstar} we study the relation between conditional $s^*$ and conditional maximal correlation.

Conditional hypercontractivity ribbon is  useful in studying tensorization for two-way channels, as recently shown by authors in~\cite{OurPaper}. We briefly discuss this in Section~\ref{sec:two-way}. Also, an application  of conditional hypercontractivity ribbon for secure distribution simulation is given in Appendix~\ref{appndx:secsim}.

\section{Computation of the regions and their local perturbation}\label{sec:computing}

Explicit computation of the tensorizing regions defined so far for a given joint distribution can be computationally cumbersome, specially for distributions defined on large alphabet sets. This computation can be relatively simplified if one observes that expressions with auxiliary random variables generally have alternative representations in terms of lower convex envelopes\footnote{A lower convex envelope of a function is the largest convex function that lies below the function.} (see e.g., \cite{ChandraEnvelope}). Consider for instance 
$$s^*(X,Y)=\sup_{U-X-Y}\frac{I(U;Y)}{I(U;X)}.$$
A representation of this quantity in terms of lower convex envelopes is given in~\cite{MC-HC}.
Indeed, $s^*(X,Y)$ can be written as the minimum value of $\lambda$ such that
\begin{align}H(Y)-\lambda H(X)\leq \min_{U: U-X-Y}\big[H(Y|U)-\lambda H(X|U)\big].
\label{eqn:sstareho1}
\end{align}
The right hand side of this equation has a representation in terms of the lower convex envelope operator as follows. Given $p(x, y)=p(x)p(y|x)$, we fix the channel $p(y|x)$ and vary the input distribution to define the following function 
$$t_{\lambda}(q(x))=H(Y)-\lambda H(X),$$
where entropies are computed with respect to $q(x,y)=q(x)p(y|x)$. Then 
$$\min_{U:U-X-Y}\big[H(Y|U)-\lambda H(X|U)\big],$$ 
is the lower convex envelope of the function
$t_{\lambda}(q(x))$ at $q(x)=p(x)$. Equation~\eqref{eqn:sstareho1} then implies that $s^*(X, Y)$ is the minimum value of
$\lambda$ such that the function $t_{\lambda}(q(x))$ touches its lower convex envelope at $p(x)$. 

The lower convex envelope operator is still a global operator. In order to further simplify the computation, one can 
replace lower convex envelopes with the weaker constraint of local convexity, i.e., to consider the minimum value of
$\lambda$ such that the function $t_{\lambda}(q(x))$ is locally convex (has a positive semi-definite Hessian) at $p(x)$. This quantity is clearly a lower bound on $s^*(X,Y)$, and is shown in~\cite{MC-HC} to be equal to $\rho(X, Y)^2$, where $\rho(X, Y)$ is the maximal correlation between $X$ and $Y$. The quantity $\rho(X, Y)$ has an efficient representation in terms of principal inertia components (see~\cite[Sec. II. B]{medard} and references therein). As discussed in the introduction it also satisfies the tensorization and data processing properties.

More generally, in~\cite{OurPaper} the local approximation of the bipartite hypercontractivity ribbon is derived and the \emph{maximal correlation ribbon} is defined. It is shown that this ribbon satisfies tensorization and data processing. One can apply this idea of local approximation to other regions defined in this paper. Here we do this for the region given in Section \ref{sec:ex2}.

In Section \ref{sec:ex2}, it was shown that  
\begin{align}G^s_{X_{[{k+1}]}}(\lambda_{[k]})=\max_{U-X_{k+1}-X_{[k]}}-I(X_{k+1};U)+\sum_{i=1}^{k}\lambda_i I(X_i;U),\label{defsecG}
\end{align}
is additive and satisfies the data processing inequality. This function can be written as
\begin{align*}
G^s_{X_{[{k+1}]}}(\lambda_{[k]})&=-H(X_{k+1})+\sum_{i=1}^{k}\lambda_i H(X_i)+\max_{U-X_{k+1}-X_{[k]}}\bigg[H(X_{k+1}|U)-\sum_{i=1}^{k}\lambda_i H(X_i|U)\bigg]\\&=-H(X_{k+1})+\sum_{i=1}^{k}\lambda_i H(X_i)-\min_{U-X_{k+1}-X_{[k]}}\bigg[-H(X_k|U)+\sum_{i=1}^{k}\lambda_i H(X_i|U)\bigg].
\end{align*}
This function is less than or equal to zero if and only if for any $p(u|x_{k+1})$ we have
$$-H(X_{k+1}|U)+\sum_{i=1}^{k}\lambda_i H(X_i|U)\geq -H(X_{k+1})+\sum_{i=1}^{k}\lambda_i H(X_i).$$
In other words, we have $\lambda_{[k]}\in \Upsilon^s(X_{[{k+1}]})$ if
the function 
\begin{align}\label{eq:map-q-h-1}
t_{\lambda_{[k]}}\big(q(x_{k+1})\big)= -H(X_{k+1})+\sum_{i=1}^{k}\lambda_i H(X_i),
\end{align}
when we fix $p(x_{[k]}|x_{k+1})$ and vary the marginal distribution of $X_{k+1}$, lies on its lower convex envelope at $q(x_{k+1})=p(x_{k+1})$. 

Now, instead of being on the lower convex envelope, we look at the local convexity of $t_{\lambda_{[k]}}(\cdot)$ at $q(x_{k+1})=p(x_{k+1})$. Local convexity is a necessary condition for being on the lower convex envelope. To verify local convexity, consider a local perturbation of the form $q_\epsilon(x_{k+1})=p(x_{k+1})(1+\epsilon f(x_{k+1}))$. Assuming that $\mathbb{E}[f(X_{k+1})]=0$, then for sufficiently small $|\epsilon|$, this equation defines a valid distribution. Then we may consider the distribution $q_{\epsilon}(x_{[k+1]})=q_\epsilon(x_{k+1})p(x_{[k]}|x_{k+1})$. The second derivative of~\eqref{eq:map-q-h-1} with respect to $\epsilon$ at $\epsilon=0$ is equal to~\cite{GMartonPaper}
$$\frac{\partial^2}{\partial \epsilon^2} t_{\lambda_{[k]}}\big(q_\epsilon(x_{[k+1]})\big) \Big|_{\epsilon=0}=\mathbb{E}[f(X_{k+1})^2]-\sum_{i=1}^{k}\lambda_i\mathbb{E}[\mathbb{E}[f(X_{k+1})^2|X_i]].$$
We would like this to be non-negative for all valid perturbations $f$. Then we obtain the following new region.

\begin{definition} \label{def10}
Define 
\begin{align*}
\Lambda^s(X_{[{k+1}]})&=\big\{\lambda_{[k]}\in \mathbb{R}_{+}^{k}\big|\, \mathbb{E}[f(X_{k+1})^2]\geq \sum_{i=1}^{k}\lambda_i\mathbb{E}[\mathbb{E}[f(X_{k+1})^2|X_i]], \forall f(X_{k+1}): \mathbb{E}[f(X_{k+1})]=0\big\}
\\&=\big\{\lambda_{[k]}\in \mathbb{R}_{+}^{k}\big|\, {\Var}[f(X_{k+1})]\geq \sum_{i=1}^{k}\lambda_i{\Var}_{X_i}[\mathbb{E}_{X_{k+1}|X_i}[f(X_{k+1})]],\quad \forall f(X_{k+1})\big\}.
\end{align*}
\end{definition}

The region $\Lambda^s(X_{[{k+1}]})$ again satisfies data processing and tensorization. To prove this we define the following function
\begin{align}
\widetilde{G}^s_{X_{[k+1]}}(\lambda_{[k]})&=\max_{f(X_{k+1})}\bigg[-{\Var}[f(X_{k+1})]+\sum_{i=1}^{k}\lambda_i{\Var}_{X_i}[\mathbb{E}_{X_{k+1}|X_i}[f(X_{k+1})]]\bigg].
\label{defsecG2}
\end{align}
Then the data processing and tensorization of $\Lambda^s(X_{[{k+1}]})$ is equivalent to the data processing and additivity of $\widetilde{G}^s_{X_{[k+1]}}(\lambda_{[k]})$.

Comparing equations~\eqref{defsecG} and~\eqref{defsecG2}, we see that the term $I(U;X_i)$ is replaced with 
$${\Var}_{X_i}[\mathbb{E}_{X_{k+1}|X_i}[f(X_{k+1})]].$$
This suggests that an \emph{algebraic} proof of additivity and data processing of $G^s_{X_{[{k+1}]}}$ can be mimicked to obtain a proof of these properties for $\widetilde{G}^s_{X_{[{k+1}]}}$. 
Indeed using Table~\ref{table:similarities}, we may transform any algebraic relation between quantities in terms of  mutual information, to a similar equation in terms of variance. 
In particular, the chain rule for mutual information corresponds to the law of total variance. The fourth property $I(U;C|DE)\geq I(U;C|D)$ holds for mutual information 
 since $I(U;C|DE)=I(UE;C|D)\geq I(U;C|D)$. The proof of its analogue for variance is similar and can be found in~\cite[Lemma 30]{OurPaper}. 
 Using these properties, we show in Appendix~\ref{appdx:localpertub} that a proof of additivity and data processing for $G^s_{X_{[k+1]}}$ gives a similar proof for $\widetilde{G}^s_{X_{[k+1]}}$. For another proof of this type, see the proofs of the data processing and tensorization properties of hypercontractivity ribbon and maximal correlation ribbon in~\cite{OurPaper}.

\begin{table}
\begin{center}
    \begin{tabular}{|c | c | c|}
    \hline
    &\textbf{Mutual Information} & \textbf{Variance} \\ \hline
   1& $I(U;B)$ with $U-A-B$ & ${\Var}_{B}[\mathbb{E}_{A|B}[f(A)]]$ \\ \hline
2&$I(U;C|B)$ with $U-A-BC$ & $\mathbb{E}_{B}{\Var}_{C|B}[\mathbb{E}_{A|BC}[f(A)]]$ \\ \hline
 &   Chain rule &  Law of total variance\\ 
 3&$I(U;BC)=I(U;B)+I(U;C|B)$ &${\Var}_{BC}[\mathbb{E}_{A|BC}[f(A)]]={\Var}_{B}[\mathbb{E}_{A|B}[f(A)]]\quad\qquad\qquad$ \\
&&$\qquad\qquad\qquad \qquad\quad +\mathbb{E}_{B}{\Var}_{C|B}[\mathbb{E}_{A|BC}[f(A)]]$\\
\hline
4&$I(U;C|DE)\geq I(U;C|D)$ & $\E_{DE}\Var_{C|DE}\E_{A|CDE}[f(A)] \geq \E_D \Var_{C|D} \E_{A|CD}[f(A)]$ 
\\ &if $C-D-E,~~U-A-CDE$ & if $C-D-E$
\\
\hline
    \end{tabular}
\end{center}
\label{table:similarities}\caption{Algebraic similarities between mutual information and variance}
\end{table}

%*******************************************

\section{Two-way channels}\label{sec:two-way}
So far we have only considered source coding problems. We now consider a two-way channel coding problem. Let us begin by motivating our problem. Let $p(y|x)$ and $q(\tilde{y}|\tilde{x})$ be two point-to-point channels. The question is whether we can simulate one use (copy) of the channel $q(\tilde{y}|\tilde{x})$ from arbitrarily many uses of $p(y|x)$.  In other words,  given some arbitrary small error $\epsilon>0$, can we find some $n$ and (possibly randomized) encoder $\mathcal{E}:\tilde{x}\mapsto x^n$ and decoder $\mathcal{D}:y^n\mapsto \tilde{y}$ such that the induced conditional distribution of $\tilde{y}$ given $\tilde{x}$ is within the $\epsilon$ distance of $q(\tilde{y}|\tilde{x})$ for every $\tilde{x},\tilde{y}$? This question for point-to-point channels as stated here, is easy to answer. Indeed, if the capacity of $q(\tilde y| \tilde x)$ is zero, then we only need local randomness to simulate it. Otherwise, simulation is feasible iff the capacity of $p(y|x)$ is non-zero. We observe that the answer to the simulation problem for point-to-point channels is easy since such channels with zero capacity have a trivial characterization.

Let us ask the same question for two-way channels: can we simulate a single copy of $q(\tilde{y}_1, \tilde{y}_2|\tilde{x}_1, \tilde{x}_2)$ from an arbitrary number of copies of $p(y_1, y_2|x_1, x_2)$? More precisely, is there $n$ and local encoding maps $\mathcal E_i: \tilde x_i \mapsto x_i^n$, for $i=1, 2$ and decoding maps $\mathcal D_i: y_i^n\mapsto \tilde y_i$ such that the induced conditional distribution of $(\tilde y_1, \tilde y_2)$ conditioned on $(\tilde x_1, \tilde x_2)$ is within $\epsilon$ distance of $q(\tilde{y}_1, \tilde{y}_2|\tilde{x}_1, \tilde{x}_2)$? 

We may make this problem even more general by adding feedback to the channel. In this case the $i$-th encoder, $i=1,2$, before using the $j$-th copy of $p(y_1, y_2|x_1, x_2)$ have access to the outputs of previous channels. 
More specifically, assume that there are two parties who have the channel  $p(y_1, y_2|x_1, x_2)$ as a resource between them, which they can use arbitrarily many times. To begin with, the  $i$-th party, $i=1,2$,  is given $\tilde{x}_i$, the input of the channel  to be simulated. The $i$-th party creates input $X_{ij}$ at time instance $j$, using his past inputs and outputs of the channel, i.e., from $(\tilde{x}_i, X_{i[j-1]}, Y_{i[j-1]})$. After feeding $(X_{1j}, X_{2j})$ to the $j$-th copy of $p(y_1, y_2|x_1, x_2)$, the output $(Y_{1j}, Y_{2j})$ is generated. Finally, after using the two-way channel $p(y_1, y_2|x_1, x_2)$ for $n$ times, the $i$-th party creates $\tilde Y_i$ from $(\tilde{x}_i, X_{i[n]}, Y_{i[n]})$ to create $\tilde{Y}_i$. We need the imposed conditional distribution on $\tilde{Y}_1, \tilde{Y}_2$ to be close to $q(\tilde{y}_1, \tilde{y}_2|\tilde{x}_1, \tilde{x}_2)$.
 
To answer the possibility of channel simulation in the bipartite case as above, it is appropriate to restrict ourselves to zero-capacity channels (i.e., to channel whose capacity region is $\mathcal{C}=\{(0,0\}$). The point is that (unlike the point-to-point case) there are non-trivial two-way channels with zero capacity.

%The answer is negative! This problem has been raised in our recent work in \cite{OurPaper}.  

Consider for instance, the following class of zero-capacity channels with binary inputs and binary outputs (i.e., $y_1, y_2, x_1, x_2\in \{0,1\}$): 
\begin{align}\label{eq:iso-box}
\PR_{\eta}(y_1,y_2|x_1,x_2) := \begin{cases}
\frac{1+\eta}{4} \qquad\qquad \text{if } y_1 \oplus y_2= x_1\wedge x_2,\\
\frac{1-\eta}{4} \qquad\qquad \text{otherwise},\\
\end{cases}
\end{align}
where  $0\leq \eta\leq 1$. 
Then the following statement is proved in our recent work~\cite{OurPaper}.
\begin{theorem} \cite{OurPaper}
For $1/2< \eta_1<\eta_2<1$, two parties cannot use an arbitrary number of copies of $\PR_{\eta_1}$ to  generate a single copy of $\PR_{\eta_2}$.\label{thm:PR}
\end{theorem}

Our goal here is to illustrate this result from the perspective of additivity and tensorization, based on the ideas we developed.  
%In particular, we define a function $G$ that is additive. The notion of additivity for problems with feedback or correlated sources has been formalized in \cite{GeneralAdditivity} using the notion of \emph{information state}, i.e. all the information that is available to the parties at a given time instance. Then the additivity constraint is a one-step condition. 

Given $p(y_1, y_2)$, 
define 
\begin{align}\label{eq:g-z-def}
G^z_{\lambda_1, \lambda_2}(Y_1, Y_2)=\max_{p(u|y_1y_2)}-I(U;Y_1Y_2)+\lambda_1 I(U;Y_1)+\lambda_2 I(U;Y_2).
\end{align}
Observe that this function is the one for bipartite hypercontractivity ribbon and is a special case of~\eqref{defGinV6}. Therefore, it satisfies the data processing and additivity properties. Now, given a two-way channel $q(y_1, y_2|x_1, x_2)$, let
$$G^z_{\lambda_1, \lambda_2}(q(y_1,y_2|x_1, x_2))=\max_{x_1, x_2}{G}^z_{\lambda_1, \lambda_2}(Y_1, Y_2|X_1=x_1, X_2=x_2).$$
Observe that $G^z_{\lambda_1, \lambda_2}(q(y_1,y_2|x_1, x_2))$ indeed corresponds to the conditional hypercontractivity ribbon of outputs given inputs, as in Lemma \ref{lemma9}.
The following lemma is the key step to prove Theorem~\ref{thm:PR}.

\begin{lemma}\label{lemma12}
Assume that $(A, B)$ are sampled from some bipartite distribution $p(a,b)$. Suppose that we create $X_1$ as a \emph{function} of $A$, and $X_2$ as a \emph{function} of $B$. Then $(X_1, X_2)$ are put at the inputs of a two-way channe $p(y_1,y_2|x_1, x_2)$ which outputs $(Y_1, Y_2)$. Then for any $\lambda_1, \lambda_2\geq 0$ we have
$$G^z_{\lambda_1, \lambda_2}(AY_1, BY_2) - \lambda_1 I(X_2; Y_1|X_1)-\lambda_2 I(X_1;Y_2|X_2)\leq 
G^z_{\lambda_1, \lambda_2}(A,B)+G^z_{\lambda_1, \lambda_2}(p(y_1,y_2|x_1, x_2)).$$
\end{lemma}

Assuming this lemma the following theorem gives a method for proving the impossibility of channel simulation.

\begin{theorem} For any two-way channel $p(x_1, x_2| y_1, y_2)$ let 
$$\Upsilon^z(p(x_1, x_2| y_1, y_2)) = \{(\lambda_1, \lambda_1)\in [0, 1]^2|\,  G^z_{\lambda_1, \lambda_2}(q(y_1,y_2|x_1, x_2))\leq 0\}.$$
Assume that $p(x_1, x_2| y_1, y_2)$ has zero capacity. Then simulation of $q(\tilde{y}_1, \tilde{y}_2|\tilde{x}_1, \tilde{x}_2)$ with $p(x_1, x_2| y_1, y_2)$, as defined above, is possible only if $\Upsilon^z(p(x_1, x_2| y_1, y_2))\subseteq \Upsilon^z(q(\tilde{y}_1, \tilde{y}_2|\tilde{x}_1, \tilde{x}_2))$.
\end{theorem}

\begin{proof}
Let $A$ and $B$ respectively denote all information available to the two parties (including their private randomness) before using the two-way channel $p(y_1,y_2|x_1, x_2)$ at some time step.  Then their available information after using the channel is $AY_1$ and $BY_2$. When the channel has zero capacity, we have $I(X_2; Y_1|X_1=x_1)=0$ for every value of $x_1$, \cite[Proposition 17.2]{elgamal}; similarly, we have $I(X_1;Y_2|X_2=x_2)=0$ for every value of $x_2$. Thus, $I(X_2; Y_1|X_1)=I(X_1;Y_2|X_2)=0$ for any $p(x_1,x_2)$. Then by Lemma~\ref{lemma12} we have 
$$G^z_{\lambda_1, \lambda_2}(AY_1, BY_2) \leq 
G^z_{\lambda_1, \lambda_2}(A,B)+G^z_{\lambda_1, \lambda_2}(p(y_1,y_2|x_1, x_2)).$$
This means that, if $G^z_{\lambda_1, \lambda_2}(A,B)\leq 0$ and $G^z_{\lambda_1, \lambda_2}(p(y_1,y_2|x_1, x_2))\leq 0$, then 
$G^z_{\lambda_1, \lambda_2}(AY_1, BY_2) \leq 0$. 

Now consider a simulation code with error $\epsilon$. The initial information state is $(T'_1,T'_2)=(\tilde{x}_1T_1,\tilde{x}_2T_2)$, where $\tilde{x}_1$ and $\tilde{x}_2$ are two constants (the inputs of the channel we want to simulate), and $T_1$ and $T_2$ are two mutually independent private sources of randomness. Since $T'_1, T'_2$ are independent for any $(\lambda_1, \lambda_2)\in [0,1]^2$ we have $G^z_{\lambda_1, \lambda_2}(T'_1,T'_2)= 0$. Therefore, if $(\lambda_1, \lambda_2)$ is such that $G^z_{\lambda_1, \lambda_2}(p(y_1,y_2|x_1, x_2))\leq 0$, by repeating the above argument we find that $$G^z_{\lambda_1, \lambda_2}(\tilde{x}_1X_{1[n]}Y_{1[n]}, \tilde{x}_2X_{2[n]}Y_{2[n]})\leq 0$$ at the final stage of communication. From the data processing property of $G^z_{\lambda_1, \lambda_2}$, we find that $G^z_{\lambda_1, \lambda_2}(\tilde{Y}_1, \tilde{Y}_2)\leq 0$.  Thus, for any arbitrary $p(u|\tilde y_1\tilde y_2)$, we have 
$$-I(U;\tilde Y_1\tilde Y_2)+\lambda_1 I(U;\tilde Y_1)+\lambda_2 I(U;\tilde Y_2)\leq 0.$$
Now, letting $\epsilon$ converge to zero and using the continuity of mutual information in the underlying distribution, we get that $(\lambda_1, \lambda_2)$ belongs to $\Upsilon^z(q(\tilde{y}_1, \tilde{y}_2|\tilde{x}_1, \tilde{x}_2))$.
This gives the desired result. 
\end{proof}

We now give a proof for Lemma~\ref{lemma12}.

\begin{proof}[Proof of Lemma \ref{lemma12}]
Take some $p(u, a, b, x_1, x_2, y_1, y_2)$ that achieves the maximum in $G^z_{\lambda_1, \lambda_2}(AY_1, BY_2)$. Then we have
\begin{align}I(U;Y_1Y_2AB)&=I(U;AB)+I(U;Y_1Y_2|AB)\nonumber
\\&=I(U;AB)+I(U;Y_1Y_2|ABX_1X_2)\label{eqn:Ltwo-way1}
\\&=I(U;AB)+I(UAB;Y_1Y_2|X_1X_2)\nonumber
\\&=I(U;AB)-\lambda_1 I(U;A)-\lambda_2I(U;B) \nonumber
\\&\quad+I(UAB;Y_1Y_2|X_1X_2)-\lambda_1 I(UAB;Y_1|X_1X_2)-\lambda_2 I(UAB;Y_2|X_1X_2)\nonumber
\\&\quad +\lambda_1 I(U;A)+\lambda_2I(U;B) +\lambda_1 I(UAB;Y_1|X_1X_2)+\lambda_2 I(UAB;Y_2|X_1X_2)\nonumber
\\&\geq -G^z_{\lambda_1, \lambda_2}(A,B)-G^z_{\lambda_1, \lambda_2}(q(y_1,y_2|x_1, x_2))\nonumber
\\&\quad +\lambda_1 I(U;A)+\lambda_2I(U;B) +\lambda_1 I(UAB;Y_1|X_1X_2)+\lambda_2 I(UAB;Y_2|X_1X_2),\nonumber
\end{align}
where equation \eqref{eqn:Ltwo-way1} follows from the fact that $X_1$ and $X_2$ are functions of $A$ and $B$ respectively.

Since $p(u, a, b, x_1, x_2, y_1, y_2)$ achieves the maximum in $G^z_{\lambda_1, \lambda_2}(AY_1, BY_2)$, we have $$-G^z_{\lambda_1, \lambda_2}(AY_1, BY_2)=I(U;Y_1Y_2AB)-\lambda_1 I(U;Y_1A)-\lambda_2 I(U;Y_2B).$$
Hence,
\begin{align*}
-G^z_{\lambda_1, \lambda_2}(AY_1, BY_2)&\geq 
-G^z_{\lambda_1, \lambda_2}(A,B)-G^z_{\lambda_1, \lambda_2}(q(y_1,y_2|x_1, x_2))
\\&\quad +\lambda_1 I(U;A)+\lambda_2I(U;B) +\lambda_1 I(UAB;Y_1|X_1X_2)+\lambda_2 I(UAB;Y_2|X_1X_2)
\\&\quad -\lambda_1 I(U;Y_1A)-\lambda_2 I(U;Y_2B)
\\&= 
-G^z_{\lambda_1, \lambda_2}(A,B)-G^z_{\lambda_1, \lambda_2}(q(y_1,y_2|x_1, x_2))
\\&\quad+ \lambda_1 [I(UAB;Y_1|X_1X_2)-I(U;Y_1|A)]+\lambda_2 [I(UAB;Y_2|X_1X_2)-I(U;Y_2|B)]
\\&= 
-G^z_{\lambda_1, \lambda_2}(A,B)-G^z_{\lambda_1, \lambda_2}(q(y_1,y_2|x_1, x_2))
\\&\quad+ \lambda_1 [I(UAB;Y_1|X_1X_2)-I(U;Y_1|AX_1)]+\lambda_2 [I(UAB;Y_2|X_1X_2)-I(U;Y_2|BX_2)]
\\&= 
-G^z_{\lambda_1, \lambda_2}(A,B)-G^z_{\lambda_1, \lambda_2}(q(y_1,y_2|x_1, x_2))
\\&\quad+ \lambda_1 [I(UABX_2;Y_1|X_1)-I(U;Y_1|AX_1)]+\lambda_2 [I(UABX_1;Y_2|X_2)-I(U;Y_2|BX_2)]
\\&\quad- \lambda_1 I(X_2;Y_1|X_1)-\lambda_2 I(X_1;Y_2|X_2)
\\&\geq
-G^z_{\lambda_1, \lambda_2}(A,B)-G^z_{\lambda_1, \lambda_2}(q(y_1,y_2|x_1, x_2))
\\&\quad- \lambda_1 I(X_2;Y_1|X_1)-\lambda_2 I(X_1;Y_2|X_2).
\end{align*}
\end{proof}

%************************************

\section{Conclusion and Future Work}
In this paper we defined new classes of measures of correlation that satisfy the tensorization property. These measures were defined using additive functions, which themselves are useful for the non-interactive distribution simulation with a non-zero rate. Conditional versions of the proposed measures are derived, and are shown to be applicable to the secure distribution simulation problem. Since explicit computation of the proposed regions is generally difficult, we looked at local perturbation of the regions. Tensorization and data processing of the local regions can be shown via an analogy between propoerties of mutual information and variance. In the appendices, we study different characterizations of the multi-partite HC ribbon. We also define a new multi-partite maximal correlation. 

All the source coding problems that we considered have a capacity region characterized by a single auxiliary random variable. It would be interesting to consider problems with more than one auxiliary random variable. Except for the section on two-way channels, our main emphasis was on the source coding problems. It would be interesting to explore tensorizing measures for channels. 

The multi-partite HC ribbon has a description in terms of Schatten norms. For this reason, it has found applications in other areas of mathematics. It would be interesting to see whether other regions defined in this paper have similar characterizations. 

Finally, we defined a notion of multi-partite maximal correlation. It would be interesting to see if this measure is related to the maximal correlation ribbon (MC ribbon). The MC ribbon is the local perturbation of the HC ribbon, and can be derived by setting $X_{k+1}=X_{[k]}$ in Definition \ref{def10}. 

\appendix
\begin{center}\LARGE{Appendix}\end{center}

%*********************************************

%*********************************************

\section{Conditional $\rho$ and $s^*$} \label{conditionalRhoSstar}
We need the following definition:
\begin{definition} [Conditional Maximal Correlation]\cite{BeigiTse}
For a tripartite distribution $p(x,y,z)$, the conditional maximal correlation $\rho(X, Y|Z)$ is defined as
\begin{align*}
\rho(X, Y|Z)&=\max_{z:p(z)>0} \rho(X, Y|Z=z).
\end{align*}
\end{definition}

\begin{lemma}\cite{BeigiTse} We have
$$\rho^2(X, Y|Z)=\max_{\E[f|Z]=0,\, \E[f^2]=1} \E_{YZ}[(\E_{X|YZ}[f(X,Z)])^2],
$$
where maximum is taken over all functions $f:\mathcal X\times \mathcal Z\rightarrow \mathbb R$.
\end{lemma}

\begin{proof}
Here we briefly explain the idea of the proof. We first note that
\begin{align*}
\rho(X, Y|Z) =  \max  ~& \E[f{(X,Z)}\,g{(Y,Z)}]\\
&\E_{X|Z} [f] = \E_{Y|Z} [g]=0, \\
& \E [f^2]= \E [g^2] =1.
\end{align*}
To verify this, it suffices to expand the expectations $\E[\cdot]$ as $\E_{Z}[\E_{XY|Z} [\cdot]]$, and instead of functions $f(X,Z), g(Y,Z)$ to consider pairs of functions $(f(X, z), g(Y, z))$ for all $z$. 

Now having the above characterization of conditional maximal correlation we can prove the lemma. The point is that if we fix $f(x,z)$, by the Cauchy-Schwarz inequality, the optimal $g(Y,Z)$ will be proportional to $\E_{X|YZ}[f(X,Z)]$.

\end{proof}

From this definition we have $\rho(X, Y|Z)^2\leq s^*(X, Y|Z)$ since $\rho(X,Y|Z=z)^2\leq s^*(X,Y|Z=z)$ for every $z$ with $p(z)>0$. Before stating another connection between $s^*$ and $\rho$, we need the following alternative characterization of $s^*(X,Y|Z)$.

\begin{lemma}We have $$s^*(X,Y|Z)=\sup_{U:~U-XZ-Y,  I(U;Z)=0}\frac{I(U;Y|Z)}{I(U;X|Z)}.$$
In other words, in the definition conditional $s^*$ the supremum with, or without the constraint $I(U;Z)=0$ gives rise to the same value.
\end{lemma}

\begin{proof}
Take some $p(u|x,z)$, so that $U-XZ-Y$ forms a Markov chain. By the functional representation lemma \cite[Appendix B]{elgamal} applied to $p(u|z)$, one can find $p(u,u',z)$ where $U'$ is independent of $Z$ and $H(U|U'Z)=0$. Next, define the joint distribution
$$p(u',u,x,y,z)=p(u',z)p(u|u',z)p(x|u,z)p(y|x,z),$$
whose marginal distribution on $(U, Z, X, Y)$ is the one we started with. 
Observe that we have Markov chains $U'-XZ-Y$ and $U'-UZ-XY$.  Then $I(U;Y|Z)=I(U';Y|Z)$ and $I(U;X|Z)=I(U';X|Z)$. Hence,
$$\frac{I(U;Y|Z)}{I(U;X|Z)}=\frac{I(U';Y|Z)}{I(U';X|Z)},$$
and we have $I(U'; Z)=0$.
\end{proof}

We are now ready to provide an alternative characterization of conditional $\rho$ and $s^*$ in terms of lower convex envelopes. This generalizes such a characterization of~\cite{MC-HC} to the conditional case. 
%The proof is almost identical to its special case given in \cite{MC-HC}. 

 Fix $p(z)$ and a channel $p(y|xz)$. Then for $\lambda \in [0,1]$ define the following function of $p(x|z)$:
$$ t_\lambda(p(x|z)) = H(Y|Z) - \lambda H(X|Z).$$

\begin{theorem}The following statements hold:
\begin{enumerate}
\item[{\rm(i)}] $\rho^2(X, Y|Z)$ is the minimum value of $\lambda$ such that the function $t_\lambda$ 
has a positive semidefinite Hessian at $p(x|z)$. 
\item[{\rm(ii)}] $s^*(X,Y|Z)$ is the minimum value of $\lambda$ such that the function $t_\lambda$ touches its lower convex envelope at $p(x|z)$.
\end{enumerate}
\end{theorem}

\begin{proof}
(i) This follows from the following characterization of conditional maximal correlation:
$$\rho(X,Y|Z)=\max_{\E[f_{XZ}|Z]=0, \E[f^2]=1} \E_{YZ}[(\E_{X|YZ}[f(X,Z)])^2].$$
Take an arbitrary perturbation of the form
$p_{\epsilon}(x,z)=p(x,z)(1+\epsilon f(x,z))$ such that $p_{\epsilon}(z)=p(z)$. For $p_{\epsilon}$ to stay a
valid perturbation we need $\E[f]=0$, and for it to satisfy $p_{\epsilon}(z)=p(z)$, we need $\E[f|Z]=0$. Furthermore, we can normalize $f$ by assuming that $\E[f^2]=1$. With these constraints we obtain a conditional distribution $p_{\epsilon}(x|z)$ for sufficiently small $|\epsilon|$. Then we have
\begin{align*}
\frac{\partial^2}{\partial \epsilon^2}t_\lambda(p_{\epsilon}(x|z))\Big|_{\epsilon=0} & =  -\E[\E[f(X,Z)|YZ]^2]+\lambda \E[f^2(X,Z)]\\
& =-\E[\E[f(X,Z)|YZ]^2]+\lambda,
\end{align*}
which is non-negative as long as $\lambda \geq \E[\E[f(X,Z)|YZ]^2]$. Thus  the minimum value $\lambda^*$ such that the second derivative is non-negative for all local perturbations is 
$$\lambda^* =\max_{\E[f_{XZ}|Z]=0, \E[f^2]=1} \E_{YZ}[(\E_{X|YZ}f(X,Z))^2].$$

(ii) Consider the minimum value of $\lambda$, say $\tilde \lambda$, such that the function $t_\lambda$ touches its lower convex envelope at $p(x|z)$. This means that $\tilde \lambda$ is the minimum $\lambda$ such that
$$ H(Y|Z)-\lambda H(X|Z)\leq H(Y|UZ)-\lambda H(X|UZ),\qquad \forall \ U: U-XZ-Y, I(U;Z)=0.$$ 
Note that if $U$ is conditionally independent of $X$, i.e., $I(U;X|Z) = 0$, then the above inequality always holds. So let us further assume that $I(U; X|Z)>0$. Then rewriting the above equation, we find that $\tilde \lambda$ is the minimum $\lambda$ such that,
$$\lambda\geq \frac{I(U;Y|Z)}{I(U;X|Z)}, \qquad  \forall \ U: U-XZ-Y~\mbox{with}~ I(U;Z)=0, I(U;X|Z) > 0.$$
Thus, 
$$\tilde \lambda = \sup_{U:~U-XZ-Y, I(U;Z)=0}\frac{I(U;Y|Z)}{I(U;X|Z)}.$$
\end{proof}

%***************************************************

\section{Secure distribution simulation: an application of conditional hypercontractivity ribbon}
\label{appndx:secsim}

Consider two parties and an adversary who observe i.i.d.\ repetitions of $X_1$ and $X_2$ and $Z$ respectively. The goal of the parties is to securely generate a \emph{single copy} of $(Y_1, Y_2)$ with a given distribution $q(y_1, y_2)$ under local stochastic maps. More precisely we say that secure non-interactive simulation of $(Y_1, Y_2)$ from i.i.d.\ repetitions of $(X_1, X_2, Z)$ is possible if for every $\epsilon>0$ there is $n$
 such that the parties can generate \emph{a single copy} of $\hat Y_1$ and $\hat Y_2$ as stochastic functions of $X_1^n$ and $X_2^n$ respectively such that 
\begin{itemize}
\item\emph{Reliability constraint:} $(\hat Y_1,\hat Y_2)$ has a desired joint distribution $q(y_1,y_2)$, i.e., the joint distribution of the simulated random variables $p(\hat y_1,\hat y_2)$ is $\epsilon$-close to $q(y_1, y_2)$:
$$\|p(\hat y_1, \hat y_2)-q(y_1,y_2)\|_1\leq \epsilon.$$
\item\emph{Security:} $(\hat Y_1,\hat Y_2)$ is almost independent of $Z^n$:
$$I(\hat Y_1\hat Y_2;Z^n)\leq \epsilon.$$
\end{itemize}

This following theorem gives a bound on the problem of secure distribution simulation based on conditional hypercontractivity ribbon.

\begin{theorem} 
If secure distribution simulation is possible, then we have
$$\fR(X_1, X_2|Z)\subseteq \fR(Y_1, Y_2).$$
\end{theorem}

\begin{proof} 
We have
\begin{align}
\fR(X_1,X_2|Z)&=\fR(X^n_1, X^n_2|Z^n)\nonumber
\\&\subseteq\fR(\hat Y_1, \hat Y_2|Z^n)\nonumber
\\&=\bigcap_{z^n:p(z^n)>0}\fR(\hat Y_1, \hat Y_2|z^n).\label{apndx:eqad3}
\end{align}
where the first equation follows from the tensorization of conditional hypercontractivity ribbon and the second equation follows from its data processing property. 

Observe that
$$I(\hat Y_1\hat Y_2; Z^n)=\sum_{z^n}p(z^n)D\big(p(\hat y_1\hat y_2|z^n)\|p(\hat y_1\hat y_2)\big)\geq 2\sum_{z^n}p(z^n)\|p(\hat y_1\hat y_2|z^n)-p(\hat y_1\hat y_2)\|^2.$$
where we use Pinsker's inequality.  Assuming that the left hand side is at most $\epsilon$, there is some $z_0^n$ such that $p(z_0^n)>0$ and 
$$\|p(\hat y_1\hat y_2|z_0^n)-p(\hat y_1\hat y_2)\|\leq \sqrt{\frac{\epsilon}{2}}.$$ 
Now using~\eqref{apndx:eqad3}, we have
$\fR(X_1,X_2|Z)\subseteq\fR(\hat Y_1,\hat Y_2|z_0^n)$. This means that, if $(\lambda_1, \lambda_2)\in\fR(X_1,X_2|Z)$, then $(\lambda_1, \lambda_2)\in\fR(\hat Y_1,\hat Y_2|z_0^n)$, i.e., for any arbitrary $p(u|\hat y_1\hat y_2)$:
\begin{align}\label{eq:hat-w-y}
\lambda_1 I(U;\hat Y_1| Z^n=z_0^n)+\lambda_2I(U;\hat Y_2| Z^n=z_0^n)\leq I(U;\hat Y_1\hat Y_2| Z^n=z_0^n).
\end{align}
On the other hand by triangle inequality
$\|p(\hat y_1\hat y_2|z_0^n)-q(y_1y_2)\|_1\leq \epsilon+\sqrt{{\epsilon}/{2}}$. Then we may use the Fannes inequality to approximate each term of~\eqref{eq:hat-w-y} by an unconditional mutual information. Indeed, as $\epsilon\rightarrow 0$ we obtain
$$\lambda_1 I(U;Y_1)+\lambda_2I(U;Y_2)\leq  I(U;Y_1Y_2).$$
Thus, $(\lambda_1, \lambda_2)\in\fR(Y_1, Y_2)$.
\end{proof}

%******************************

\section{Additivity and data processing of $\widetilde{G}^s_{X_{[k+1]}}(\lambda_{[k]})$}
\label{appdx:localpertub}

Our goal in this appendix is to prove the additivity and data processing properties of $\widetilde{G}^s_{X_{[k+1]}}(\lambda_{[k]})$ defined in~\eqref{defsecG2}. For this we first give an algebraic proof of these properties for ${G}^s_{X_{[k+1]}}(\lambda_{[k]})$ defined in~\eqref{defsecG} and then using the recipe of Table~\ref{table:similarities} we convert it to a proof for $\widetilde{G}^s_{X_{[k+1]}}(\lambda_{[k]})$.   

\subsection{Additivity}
We start by showing that $G^s_{X_{[k+1]}}(\lambda_{[k]})$ is additive. That is, if $X_{[k+1]}$ and $Y_{[k+1]}$ are independent (but not necessarily identically distributed), then $G^s_{X_{[k+1]}Y_{[k+1]}}(\lambda_{[k]})=G^s_{X_{[k+1]}}(\lambda_{[k]})+G^s_{Y_{[k+1]}}(\lambda_{[k]})$. From the definition 
\begin{align}
G^s_{X_{[k+1]}Y_{[k+1]}}(\lambda_{[k]})=\max_{U-X_{k+1}Y_{k+1}-X_{[k]}Y_{[k]}}-I(X_{k+1}Y_{k+1};U)+\sum_{i=1}^{k}\lambda_i I(X_iY_i;U),
\end{align}
it is clear that $G^s_{X_{[{k+1}]}Y_{[{k+1}]}}(\lambda_{[k]})\geq G^s_{X_{[{k+1}]}}(\lambda_{[k]})+G^s_{Y_{[{k+1}]}}(\lambda_{[k]})$ since we can take $U$ to consist of an independent pair $(U_1, U_2)$ with $U_1-X_{{k+1}}- X_{[k]}$ and $U_2-Y_{{k+1}}-Y_{[k]}$. 

To show the other direction, note that
\begin{align}
-I(X_{k+1}Y_{k+1};U)+\sum_{i=1}^{k}&\lambda_i I(X_iY_i;U)\\
=&
-I(X_{k+1};U)+\sum_{i=1}^{k}\lambda_i I(X_i;U)\nonumber
\\&-I(Y_{k+1};U|X_{k+1})+\sum_{i=1}^{k}\lambda_i I(Y_i;U|X_i)\nonumber
\\&\leq
G^s_{X_{[{k+1}]}}(\lambda_{[k]})-I(Y_{k+1};U|X_{k+1})+\sum_{i=1}^{k}\lambda_i I(Y_i;UX_{k+1}|X_i)\nonumber
\\&=
G^s_{X_{[{k+1}]}}(\lambda_{[k]})-I(Y_{k+1};UX_{k+1})+\sum_{i=1}^{k}\lambda_i I(Y_i;UX_{k+1}X_i)\label{appc36}
\\&=
G^s_{X_{[{k+1}]}}(\lambda_{[k]})-I(Y_{k+1};UX_{k+1})+\sum_{i=1}^{k}\lambda_i I(Y_i;UX_{k+1})\label{appc37}
\\&\leq  G^s_{X_{[{k+1}]}}(\lambda_{[k]})+G^s_{Y_{[{k+1}]}}(\lambda_{[k]}),\label{appc38}
\end{align}
where in \eqref{appc36} we used the fact that $X_{[{k+1}]}$ and $Y_{[{k+1}]}$ are independent; in~\eqref{appc37} we used $I(Y_i;X_i|UX_{k+1})=0$ which holds because
\begin{align*}
I(Y_i;X_i|UX_{k+1})&\leq I(Y_iY_{k+1};X_i|UX_{k+1})
\\&=I(Y_i;X_i|UX_{k+1}Y_{k+1})+I(Y_{k+1};X_i|UX_{k+1})
\\&\leq 0+I(UY_{k+1};X_i|X_{k+1})
\\&=I(Y_{k+1};X_i|X_{k+1})+I(U;X_i|X_{k+1}Y_{k+1})
\\&=0,\end{align*}
and finally in~\eqref{appc38} we used the Markov chain condition $UX_{k+1}-Y_{k+1}-Y_{[k]}$.

To show that $\widetilde{G}^s_{X_{[{k+1}]}}(\lambda_{[k]})$ is additive, we follow similar steps. We need to show that if $X_{[{k+1}]}$ and $Y_{[{k+1}]}$ are independent  (but not necessarily identically distributed), then 
$$\widetilde{G}^s_{X_{[{k+1}]}Y_{[{k+1}]}}(\lambda_{[k]})=\widetilde{G}^s_{X_{[{k+1}]}}(\lambda_{[k]})+\widetilde{G}^s_{Y_{[{k+1}]}}(\lambda_{[k]}).$$
From the definition 
\begin{align}
\widetilde{G}^s_{X_{[{k+1}]}Y_{[{k+1}]}}(\lambda_{[k]})=\max_{f(X_{k+1}Y_{k+1})}\bigg[-{\Var}[f(X_{k+1}Y_{k+1})]+\sum_{i=1}^{k}\lambda_i{\Var}_{X_iY_i}[\mathbb{E}_{X_{k+1}Y_{k+1}|X_iY_i}[f(X_{k+1}Y_{k+1})]]\bigg],
\end{align}
it is clear that $\widetilde{G}^s_{X_{[{k+1}]}Y_{[{k+1}]}}(\lambda_{[k]})\geq \widetilde{G}^s_{X_{[{k+1}]}}(\lambda_{[k]})+\widetilde{G}^s_{Y_{[{k+1}]}}(\lambda_{[k]})$ since we can take $f(X_k,Y_k)$ to consist of a pair $(f(X_{k+1}), f(Y_{k+1}))$.

To show the other direction, note that
\begin{align}
&-\Var[f(X_{k+1}Y_{k+1})]+\sum_{i=1}^{k}\lambda_i\Var_{X_i}[\mathbb{E}_{X_{k+1}Y_{k+1}|X_iY_i}[f(X_{k+1}Y_{k+1})]\nonumber
\\&=
-\Var_{X_{k+1}}\mathbb{E}_{Y_{k+1}|X_{k+1}}[f(X_{k+1}Y_{k+1})]+\sum_{i=1}^{k}\lambda_i\Var_{X_i}[\mathbb{E}_{X_{k+1}Y_{k+1}Y_i|X_i}[f(X_{k+1}Y_{k+1})]\nonumber
\\&\qquad -\mathbb{E}_{X_{k+1}}\Var_{Y_{k+1}|X_{k+1}}[f(X_{k+1}Y_{k+1})]+\sum_{i=1}^{k}\lambda_i\mathbb{E}_{X_i}\Var_{Y_i|X_i}[\mathbb{E}_{X_{k+1}Y_{k+1}|X_iY_i}[f(X_{k+1}Y_{k+1})]\nonumber
\\&\leq
\widetilde{G}^s_{X_{[{k+1}]}}(\lambda_{[k]})-\mathbb{E}_{X_{k+1}}\Var_{Y_{k+1}|X_{k+1}}[f(X_{k+1}Y_{k+1})]\nonumber\\
& \qquad +\sum_{i=1}^{k}\lambda_i\mathbb{E}_{X_i}\Var_{Y_i|X_i}[\mathbb{E}_{X_{k+1}Y_{k+1}|X_iY_i}[f(X_{k+1}Y_{k+1})]\nonumber
\\&\leq
\widetilde{G}^s_{X_{[{k+1}]}}(\lambda_{[k]})-\mathbb{E}_{X_{k+1}}\Var_{Y_{k+1}|X_{k+1}}[f(X_{k+1}Y_{k+1})]\nonumber\\
& \qquad +\sum_{i=1}^{k}\lambda_i\mathbb{E}_{X_{k+1}X_i}\Var_{Y_i|X_iX_{k+1}}[\mathbb{E}_{Y_{k+1}|X_iY_iX_{k+1}}[f(X_{k+1}Y_{k+1})]\label{eqn:towrite1}
\\&
=
\widetilde{G}^s_{X_{[{k+1}]}}(\lambda_{[k]})-\mathbb{E}_{X_{k+1}}\Var_{Y_{k+1}|X_{k+1}}[f(X_{k+1}Y_{k+1})]\nonumber\\
& \qquad +\sum_{i=1}^{k}\lambda_i\mathbb{E}_{X_{k+1}}\Var_{Y_i|X_{k+1}}[\mathbb{E}_{Y_{k+1}|Y_iX_{k+1}}[f(X_{k+1}Y_{k+1})]\label{eqn:towrite2}
\\&\leq  \widetilde{G}^s_{X_{[{k+1}]}}(\lambda_{[k]})+\widetilde G^s_{Y_{[{k+1}]}}(\lambda_{[k]}).\nonumber
\end{align}
Here equation~\eqref{eqn:towrite1} holds because of property~4 of Table~\ref{table:similarities} for the choice of
$A=(X_{k+1}, Y_{k+1})$, $C=Y_i$, $D=X_i$, $E=X_{k+1}$. The Markov chain condition that we need to verify is 
$Y_i-X_i-X_{k+1}$, which holds because $X_{[{k+1}]}$ is independent of $Y_{[{k+1}]}$; equation~\eqref{eqn:towrite2} holds because $\mathbb{E}_{Y_{k+1}|X_iY_iX_{k+1}}[f(X_{k+1}Y_{k+1})]$ is equal to  $\mathbb{E}_{Y_{k+1}|Y_iX_{k+1}}[f(X_{k+1}Y_{k+1})]$ for $(X_{k+1}, X_i)$ is independent of $(Y_{k+1}, Y_i)$.

\subsection{Data processing}
We need to show that $\widetilde G^s_{Y_{[{k+1}]}}(\lambda_{[k]})\leq \widetilde G^s_{X_{[{k+1}]}}(\lambda_{[k]})$ for every
$p(y_i|x_i)$. As before, we prove this in two stages: 

\vspace{.15in}
\noindent
\textbf{Part I ($Y_i$ is a function of $X_i$):} Let us start with an algebraic proof of data processing for $G^s_{X_{[{k+1}]}}(\lambda_{[k]})$. Take some arbitrary $p(u|y_{{k+1}})$. Define 
$$p(u, x_{[{k+1}]}, y_{[{k+1}]})=p(x_{[{k+1}]}, y_{[{k+1}]})p(u|y_{k+1}).$$
Then we have $I(Y_{k+1};U)=I(X_{k+1};U)$ and $I(Y_i;U)\leq I(X_i;U)$. Therefore,
\begin{align}
-I(Y_{k+1};U)+\sum_{i=1}^{k}\lambda_i I(Y_i;U)&\leq -I(X_{k+1};U)+\sum_{i=1}^{k}\lambda_i I(X_i;U)
\\&\leq G^s_{X_{[{k+1}]}}(\lambda_{[k]}).
\end{align}
Since this holds for any arbitrary $p(u|y_{k+1})$, we get the desired result.

The proof for $\widetilde{G}^s_{X_{[{k+1}]}}(\lambda_{[k]})$ is similar. Take some function $f(Y_{k+1})$. Then, $f(Y_{k+1})$ can be also thought of as a function of $X_{k+1}$ since $Y_{k+1}$ itself is a function of $X_{k+1}$. Next, we have
 $$\Var_{Y_i}[\mathbb{E}_{X_{k+1}|Y_i}[f(Y_{k+1})]]\leq \Var_{X_iY_i}[\mathbb{E}_{X_{k+1}|X_iY_i}[f(Y_{k+1})]]=\Var_{X_i}[\mathbb{E}_{X_{k+1}|X_i}[f(Y_{k+1})]],$$
where the inequality follows from the law of total variance (property~3 of Table~\ref{table:similarities}).
Then, we have
\begin{align*}
-\Var[f(Y_{k+1})]+\sum_{i=1}^{k}\lambda_i\Var_{Y_i}&[\mathbb{E}_{Y_{k+1}|Y_i}[f(Y_{k+1})]]\\
&\leq -\Var[f(Y_{k+1})]+\sum_{i=1}^{k}\lambda_i\Var_{X_i}[\mathbb{E}_{X_{k+1}|X_i}[f(Y_{k+1})]]
\\&\leq \widetilde G^s_{X_{[{k+1}]}}(\lambda_{[k]}).
\end{align*}
Since this holds for any arbitrary function $f(Y_{k+1})$, we get the desired result.

\vspace{.15in}
\noindent
\textbf{Part II ($Y_i=(X_i, A_i)$ where $A_i$'s are mutually independent of each other, and of $Y_{[{k+1}]}$):} 
We would like to show that 
$$\widetilde{G}^s_{X_{[{k+1}]}}(\lambda_{[k]})= \widetilde{G}^s_{A_{[{k+1}]}X_{[{k+1}]}}(\lambda_{[k]}).$$ 
From the additivity of $\widetilde{G}^s$ for product of independent distributions we have
$\widetilde{G}^s_{A_{[{k+1}]}X_{[{k+1}]}}(\lambda_{[k]})=\widetilde{G}^s_{X_{[{k+1}]}}(\lambda_{[k]})+\widetilde{G}^s_{A_{[{k+1}]}}(\lambda_{[k]})$. Therefore, we need to show that
$$\widetilde{G}^s_{A_{[{k+1}]}}(\lambda_{[k]})=0,$$
when $A_i$'s are mutually independent.

As before let us begin with the proof of  $G^s_{A_{[{k+1}]}}(\lambda_{[k]})=0$. We need to show that for any arbitrary $p(u|a_{k+1})$ we have
\begin{align}
-I(A_{k+1};U)+\sum_{i=1}^{k}\lambda_i I(A_i;U)\leq 0.
\end{align}
This inequality holds because $I(A_i;U)=0$ for $i\in [k]$.  

Now, to show that $\widetilde{G}^s_{A_{[{k+1}]}}(\lambda_{[k]})=0$, we need to show that for any function $f(A_{k+1})$ we have
\begin{align*}
-\Var[f(X_{k+1})]+\sum_{i=1}^{k}\lambda_i\Var_{A_i}[\mathbb{E}_{A_{k+1}|A_i}[f(A_{k+1})]]\leq 0.
\end{align*}
From the independence of $A_i$ and $A_{k+1}$ we have that $\mathbb{E}_{A_{k+1}|A_i}[f(A_{k+1})]=0$. Hence, the above equation holds.

%*****************************************************************************

%*****************************************************************************

\end{document}